\documentclass[letterpaper,11pt]{article}
\usepackage[margin=1in]{geometry}
\usepackage{amsmath,amsthm,amssymb,amsfonts}
\usepackage[colorlinks=true,citecolor=blue,linkcolor=blue]{hyperref}
\usepackage[capitalize]{cleveref}
\usepackage{enumitem}
 \usepackage{multirow}
 \usepackage{hhline}
 \usepackage[table,xcdraw]{xcolor}
\usepackage{color}
\usepackage{colortbl}
\usepackage{bbm}
\usepackage[subrefformat=parens]{subcaption}
\usepackage[normalem]{ulem}
\usepackage{xspace}
\usepackage{stmaryrd}
\usepackage{todonotes}
\usepackage{ifthen}
\usepackage{comment}
\usepackage{tikz}
\usetikzlibrary{intersections, calc, arrows}
\usetikzlibrary{positioning}
\usetikzlibrary{calc}

\newcommand{\Nat}{\mathbb{N}} 
\newcommand{\Real}{\mathbb{R}}

\renewcommand\Pr{\mathop{\mathbf{Pr}}}
\newcommand\E{\mathop{\mathbf{E}}}
\newcommand{\Var}{\mathop{\mathbf{Var}}}
\def\defeq{\mathrel{\mathop:}=}
\newcommand{\thit}{t_{\mathrm{hit}}}
\newcommand{\tmix}{t_{\mathrm{mix}}}
\newcommand{\tcov}{t_{\mathrm{cov}}}
\newcommand{\tcoal}{t_{\mathrm{coal}}}
\newcommand{\tmeet}{t_{\mathrm{meet}}}
\newcommand{\tsep}{t_{\mathrm{sep}}}
\newcommand{\trel}{t_{\mathrm{rel}}}
\newcommand{\thitmax}{t_{\mathrm{HIT}}}
\newcommand{\trelmax}{t_{\mathrm{REL}}}
\newcommand{\taucons}{\tau_{\mathrm{cons}}}
\newcommand{\tcons}{t_{\mathrm{cons}}}

\newcommand{\Pseq}{\mathcal{P}}

\newcommand{\tauhit}{\tau_{\mathrm{hit}}}
\newcommand{\taumeet}{\tau_{\mathrm{meet}}}
\newcommand{\taucov}{\tau_{\mathrm{cov}}}
\newcommand{\taucoal}{\tau_{\mathrm{coal}}}

\newcommand{\Co}{\mathrm{C}}

\newcommand{\LSimP}{P_\mathrm{LS}}

\newcommand{\LMWP}{P_\mathrm{LM}}

\newcommand{\MDW}{P_\mathrm{DM}}

\newtheorem{theorem}{Theorem}[section]
\newtheorem{lemma}[theorem]{Lemma}
\newtheorem{corollary}[theorem]{Corollary}

\newtheorem{proposition}[theorem]{Proposition}

\crefname{equation}{}{}
\crefname{assumption}{Assumption}{Assumption}
\crefname{figure}{Figure}{Figure}

\title{Reversible Random Walks on Dynamic Graphs}
\author{
Nobutaka Shimizu\thanks{Tokyo Institute of Technology, Japan. {\ttfamily shimizu.n.ah@m.titech.ac.jp}} \and Takeharu Shiraga\thanks{Tokyo Institute of Technology, Japan. {\ttfamily shiraga.t.ab@m.titech.ac.jp}}
}

\begin{document}
\maketitle 
\pagestyle{plain}
\baselineskip 17pt plus 1pt minus 1pt 

\begin{abstract}
Recently, random walks on \emph{dynamic} graphs have been studied because of their adaptivity to the time-varying structure of real-world networks.
In general, there is a tremendous gap between static and dynamic graph settings for the lazy simple random walk:
Although $O(n^3)$ cover time was shown for any static graphs of $n$ vertices, there is an edge-changing dynamic graph with an exponential hitting time.
On the other hand, previous works indicate that the random walk on a dynamic graph with a time-homogeneous stationary distribution behaves almost identically to that on a static graph.
For example, the lazy simple random walk on a dynamic regular graph has an $O(n^2)$ hitting time, which is the same order as that on a static regular graph.

In this paper, we strengthen this insight by obtaining general and improved bounds.
Specifically, we consider a random walk according to a sequence $(P_t)_{t\geq 1}$ of irreducible and reversible transition matrices such that all $P_t$ have the same stationary distribution.
We bound the mixing, hitting, and cover times in terms of the hitting and relaxation times of the random walk according to the worst fixed $P_t$.
Moreover, we obtain the first bounds of the hitting and cover times of multiple random walks and the coalescing time on dynamic graphs.
These bounds can be seen as an extension of the well-known bounds of random walks on static graphs.
Our results generalize the previous upper bounds for specific random walks on dynamic graphs, e.g., lazy simple random walks and $d_{\max}$-lazy walks, and give improved and tight upper bounds in various cases.
As an interesting consequence of our generalization, we obtain tight bounds for the lazy Metropolis walk [Nonaka, Ono, Sadakane, and Yamashita, TCS10] on any dynamic graph: $O(n^2)$ mixing time, $O(n^2)$ hitting time, and $O(n^2\log n)$ cover time.
Additionally, our coalescing time bound implies the consensus time bound of the pull voting on a dynamic graph.
To obtain this bound, we establish a duality-like relation of the pull voting process and coalescing random walk on dynamic graphs, which is of independent interest.

\ 

\end{abstract}
\newpage

\tableofcontents

\newpage

\section{Introduction}
A random walk is a fundamental stochastic process on an undirected graph.
A walker starts from a specific vertex of a graph.
At each discrete time step,
a walker moves to a random neighbor.
The probability that the walker moves from $u$ to $v$ is given by $P(u,v)$,
    where the matrix $P\in[0,1]^{V\times V}$
    is called \emph{transition matrix}.
Because of their locality, simplicity, and low memory overhead, random walks have a wide range
of applications including network analysis, computational complexity, and distributed algorithms
\cite{Cooper11,HP01}.
The efficiency of a random walk can be
    measured by the rate of
    diffusion,
    which has been formalized
     by several notions including \emph{mixing time}, \emph{hitting time} and \emph{cover time}.
The mixing time is the
time for the distribution of the walker to converge to some limit distribution (called \emph{stationary distribution}).
The hitting time is the maximum expected time of the walker to visit a target vertex where the maximum is taken over the starting and target vertex.
The cover time is the expected time of the walker to visit all vertices starting from the worst vertex.
The mixing, hitting, and cover times on static graphs have been extensively studied for several decades~\cite{Aleliunas79,KLNS89,Lovasz93,BW90,Feige95up,Matthews88}.
For example, Aleliunas, Karp, Lipton, Lov\'asz, and Rackoff~\cite{Aleliunas79} proved that the cover time of the simple random walk on any $n$-vertex connected graph is $O(n^3)$.

Recently, there is a
growing interest in
    a random walk
    on a \emph{dynamic} graph since real-world networks change their structure over time~\cite{AKL18,SZ19,LMS18,CSZ20,CF03,KSS21,OT11}.
In this setting,
    at the beginning of the $t$-th round,
    the walker moves to a random neighbor on the current graph
    and then the edge set of the graph changes (we assume that the vertex set is static).
A central interest
    is the gap
    between random walks
    on a dynamic graph and a static one.
Indeed,
    while the (lazy) simple random walk has an $O(n^3)$ hitting time for any $n$-vertex static graphs,
    there is a sequence $(G_t)_{t\geq 1}$
    of connected graphs called the
    \emph{Sisyphus wheel} (\cref{fig:sisyphus}) 
    on which the hitting time
    is exponential~\cite{AKL18}.
    
On the other hand,
    several researchers
    observed that
    random walks on dynamic graphs
    behave almost 
    identically
    to that on static graphs if
    the stationary distribution of a random walk does not
    change over time~\cite{AKL18,SZ19,DR14,SZ07}.
Avin, Kouck\'{y}, and Lotker~\cite{AKL18}
    considered a random walk called  
    \emph{$d_{\max}$-lazy walk}
which has the uniform stationary distribution on any graph.
They proved that
    the cover time of this walk on any dynamic connected graph is $O(n^5\log^2 n)$,
    which was later improved by
    Denysyuk and Rodrigues~\cite{DR14}.
Sauerwald and Zanetti~\cite{SZ19}
    considered the lazy simple
    random walk on a dynamic
    connected graph with the same time-invariant degree distribution.
They obtained tight or nearly-tight
    bounds for the mixing and hitting times.
For example, they showed that both the mixing and hitting times are $O(n^2)$ on any dynamic regular graph.
These bounds are tight up to a constant factor even on static regular graphs.
See \cref{sec:previous_work} for more details about previous works.

\begin{figure}[htbp]
\center
\begin{tikzpicture}[
           every node/.style={
			circle,
text centered, 
			}
			]
\def\D{0.6cm}
\def\S{3.5cm}
\def\d{0cm}
\node[draw] (P4)  {1};
\node[below =\D of P4,draw] (P5)  {4};
\node[left =\D of P4,draw] (P1)  {0};
\node[above =\D of P4,draw] (P2)  {3};
\node[right =\D of P4,draw] (P3)  {2};
\node[left =\D of P2] (G1) {$G_1$};
\draw[-](P4)--(P5);
\draw[-](P4)--(P1);
\draw[-](P4)--(P2);
\draw[-](P4)--(P3);
\node[right = \S of P4,draw] (P31)  {2};
\node[below =\D of P31,draw] (P51)  {4};
\node[left =\D of P31,draw] (P41)  {1};
\node[above =\D of P31,draw] (P11)  {0};
\node[right =\D of P31,draw] (P21)  {3};
\node[left =\D of P11] (G2) {$G_2$};
\draw[-](P31)--(P51);
\draw[-](P31)--(P11);
\draw[-](P31)--(P21);
\draw[-](P31)--(P41);
\node[right = \S of P31,draw] (P22)  {3};
\node[below =\D of P22,draw] (P52)  {4};
\node[left =\D of P22,draw] (P32)  {2};
\node[above =\D of P22,draw] (P42)  {1};
\node[right =\D of P22,draw] (P12)  {0};
\node[left =\D of P42] (G3) {$G_3$};
\draw[-](P22)--(P52);
\draw[-](P22)--(P12);
\draw[-](P22)--(P32);
\draw[-](P22)--(P42);
\node[right = \S of P22,draw] (P13)  {0};
\node[below =\D of P13,draw] (P53)  {4};
\node[left =\D of P13,draw] (P23)  {3};
\node[above =\D of P13,draw] (P33)  {2};
\node[right =\D of P13,draw] (P43)  {1};
\node[left =\D of P33] (G4) {$G_4$};
\draw[-](P13)--(P53);
\draw[-](P13)--(P23);
\draw[-](P13)--(P33);
\draw[-](P13)--(P43);
\end{tikzpicture}
\caption{The Sisyphus wheel of five vertices. 
The Sisyphus wheel $\mathcal{G}=(G_t)_{t\geq 1}$ is defined as follows:
For each $t\geq 1$, let $V=V(G_t)=\{0,\ldots,n-1\}$, $v(t)=t\bmod (n-1)$, and $E(G_t)=\{\{v(t),i\}:i\in V\setminus \{v(t)\}$.
The lazy simple random walk starting from the vertex $0$ of $G_1$ has to choose the self-loop for $\Omega(n)$ consecutive times in order to reach the vertex $n-1$.
Note that the hitting time of the simple random walk on the Sisyphus wheel is unbounded. 
}
\label{fig:sisyphus}
\end{figure}
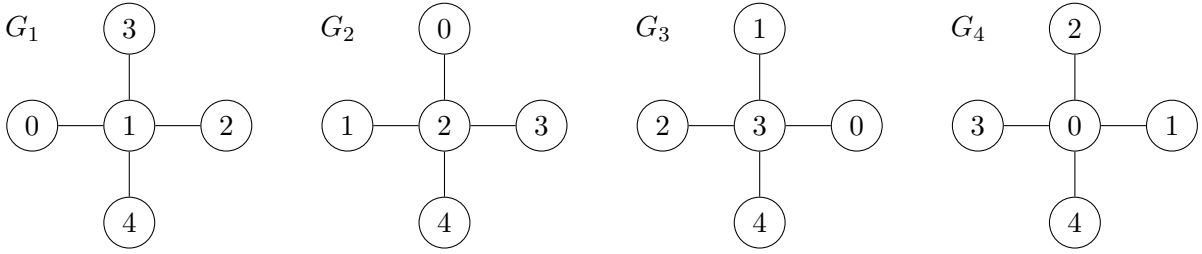

\subsection{Our results} \label{sec:our_results}
In this paper, 
we support the insight
    that a random walk
on a dynamic graph with
 the same stationary distribution
 behaves almost identically to that on a static graph
 by studying a random walk according to time \emph{inhomogeneous} transition matrices:
Given a sequence of transition matrices $\Pseq=(P_t)_{t\geq 1}$
    where $P_t\in[0,1]^{V\times V}$ for all $t$,
we consider
    a \emph{random walk according to $\Pseq=(P_t)_{t\geq 1}$} that is a sequence of random variables $(X_t)_{t\geq 0}$ satisfying $\Pr[X_t=v_t|X_{0}=v_{0},\ldots,X_{t-1}=v_{t-1}]=\Pr[X_t=v_t|X_{t-1}=v_{t-1}]=P_{t}(v_{t-1},v_t)$ for any $t\geq 1$ and $(v_0,\ldots,v_t)\in V^{t+1}$.
In other words, at the $t$-th time step ($t\geq 1$), the walker at vertex $u$ randomly selects a vertex according to the distribution $P_{t}(u,\cdot)$.
Our interest is to bound the mixing, hitting, cover, and coalescing times
under
the assumption
that all $P_t$ has
the same stationary distribution.

We briefly introduce essential terminologies to state our results.
Let $P\in [0,1]^{V\times V}$ be an \emph{irreducible} and \emph{reversible} transition matrix
and $\pi\in (0,1]^V$ be its stationary distribution\footnote{A transition matrix $P$ is \emph{irreducible} if for any $u,v\in V$ there exists a $t>0$ such that $P^t(u,v)>0$ holds
and
\emph{reversible}
if there is $\pi\in\mathbb{R}^V$ such that $\pi(u)P(u,v)=\pi(v)P(v,u)$ holds for any $u,v\in V$.
A probability distribution $\pi$ is a \emph{stationary distribution} of $P$ if $\pi P=\pi$ holds.}.
Let $\lambda_{\star}(P)$ denote the second largest eigenvalue in absolute value.
Note that $\lambda_{\star}(P)<1$ if $P$ is \emph{aperiodic}
or \emph{lazy}\footnote{We say that
 $P$ is \emph{aperiodic} if for any $v\in V$, $\mathrm{gcd}\{t\geq 0: P^t(v,v)>0\}=1$ holds
 and
$P$ is \emph{lazy} if $P(v,v)\geq 1/2$ holds for any $v\in V$.}.
Let $\trel(P)\defeq (1-\lambda_{\star}(P))^{-1}$ be the \emph{relaxation time} of $P$. 
Let $\Pseq=(P_t)_{t\geq 1}$ be a sequence of transition matrices.
For the random walk $(X_t)_{t\geq 0}$ according to $\Pseq$, 
let $\thit(\Pseq)\defeq \max_{u,w\in V}\E\left[\min\{t\geq 0: X_t=w\}\middle|X_0=u\right]$ be the worst-case expected \emph{hitting time} of $\Pseq$.
We sometimes identify $P$ with the sequence $(P_t)_{t\geq 1}$ of transition matrices with $P_t=P$ for all $t\geq 1$.
For example, $\thit(P)$ denotes the hitting time of $\mathcal{P}=(P_t)_{t\geq 1}$ with all $P_t=P$.
Let
\begin{align*}
    \trelmax(\Pseq)\defeq \max_{t\geq 1}\trel(P_t) \hspace{1em}\text{and}\hspace{1em}
    \thitmax(\Pseq)\defeq \max_{t\geq 1}\thit(P_t).
\end{align*}

\paragraph*{Mixing time.}
Our first result concerns the mixing time.
Specifically, for a sequence $\Pseq=(P_t)_{t\geq 1}$ of transition matrices, a positive probability vector $\pi\in(0,1]^V$, and a parameter $\epsilon>0$,
we define the \emph{uniform mixing time} $\tmix^{(\infty,\pi)}(\Pseq,\epsilon)$ by
\begin{align}
    \tmix^{(\infty,\pi)}(\Pseq,\epsilon)= \min\left\{t\geq 0:\max_{s\geq 0, u,v\in V}\left|\frac{(P_{s+1}P_{s+2}\cdots P_{s+t})(u,v)}{\pi(v)}-1\right|\leq \epsilon\right\}
    \label{def:uniform_mixing}
\end{align}
and let $\tmix^{(\infty,\pi)}(\Pseq)\defeq \tmix^{(\infty,\pi)}(\Pseq,1/2)$.
%
The uniform mixing time $\tmix^{(\infty)}(P,\epsilon)$ for a static Markov chain $P$ can be seen as the mixing time using the $\ell^\infty$-norm metric
and has been well studied (see, e.g., Section 4.7 in \cite{LP17}).
The intuition behind our
    definition of $\tmix^{(\infty,\pi)}(\Pseq,\epsilon)$
    is that the walk mixes well after $\tmix^{(\infty,\pi)}(\Pseq,\epsilon)$ steps
    even if the walker starts at any moment.

As a consequence of the previous work ((3.13) in \cite{SZ07}), we can easily obtain the following uniform mixing time bound.
\begin{proposition}
\label{prop:mixing_SZ07}
Let $\mathcal{P}=(P_t)_{t\geq 1}$ be a sequence of irreducible, aperiodic, and reversible transition matrices.
Suppose that all $P_t$ have the same stationary distribution $\pi$.
Then, for any $\epsilon>0$, 
$\tmix^{(\infty,\pi)}(\Pseq,\epsilon)= O\left(\trelmax(\Pseq)\log (\pi_{\min}^{-1}\epsilon^{-1})\right)$,
where $\pi_{\min}\defeq \min_{v\in V} \pi(v)$.
\end{proposition}
\Cref{prop:mixing_SZ07} can be
    seen as an extension of the well-known mixing time bound $\tmix^{(\infty)}(P,\epsilon)\leq \lceil\trel(P)\log (\pi_{\min}^{-1}\epsilon^{-1})\rceil$
    for static $P$ (Theorem 12.4 in \cite{LP17}) to the dynamic $\mathcal{P}$.

Our first result is the extension of the following well-known mixing time bound for static $P$ (Theorem 10.22 in \cite{LP17}): $\tmix^{(\infty)}(P,1/4)\leq 4\thit(P)+1$.

%
\begin{theorem}[Main result 1]
\label{thm:main_mixing}
Let $\mathcal{P}=(P_t)_{t\geq 1}$ be a sequence of irreducible, reversible, and lazy transition matrices.
Suppose that all $P_t$ have the same stationary distribution $\pi$.
Then, for any $0<\epsilon<1$, $\tmix^{(\infty,\pi)}(\Pseq,\epsilon)=O\left(\thitmax(\Pseq)+\trelmax(\Pseq)\log \epsilon^{-1}\right)$.
\end{theorem}
Note that $\trel(P)\leq \thit(P)$ holds for any irreducible, reversible and lazy $P$ (Lemma 4.24 in~\cite{AF02})
and thus \cref{thm:main_mixing} implies
$\tmix^{(\infty,\pi)}(\Pseq,1/4)=O\left(\thitmax(\Pseq)\right)$.
Compared to \cref{prop:mixing_SZ07},
    \cref{thm:main_mixing}
    eliminates the dependency
    of $\pi_{\min}^{-1}$ in the mixing time bound
    at the cost of additional $\thitmax$ term
    and the laziness assumption.
\Cref{prop:mixing_SZ07}
    gives a better bound
    if all $P_t$ has a small relaxation time (e.g., random walks on expanders).
On the other hand,
    for $\mathcal{P}$ with $\trelmax\approx \thitmax$ (e.g., lazy simple random wallk on dynamic cycles, on which both $\thitmax$ and $\trelmax$ are $\Theta(n^2)$),
    \cref{thm:main_mixing} provides a better bound.


\paragraph*{Hitting time and cover time.}
For the hitting and cover times, we recall an exponential lower bound on the Sisyphus wheel (\cref{fig:sisyphus}), which implies the following.
\begin{proposition}[\cite{AKL18}] \label{prop:sisyphus_wheel}
There is a sequence of irreducible, reversible, and lazy transition matrices $\Pseq=(P_t)_{t\geq 1}$ of $\thit(\Pseq)=2^{\Omega(n)}$.
\end{proposition}
Note that, the sequence $\mathcal{P}$ in \cref{prop:sisyphus_wheel}
    has a time-varying stationary distribution.
Our second result concerns the hitting and cover times of \emph{multiple} random walks according to $\Pseq=(P_t)_{t\geq 1}$, where all $P_t$ have the same stationary distribution.
For $k\in\Nat$,
    let $\thit^{(k)}(\Pseq)$, $\tcov^{(k)}(\Pseq)$
    denote the worst-case expected hitting and cover times
    of $k$ independent random walks
    each is according to $\Pseq$ (see \cref{sec:mixing_time_definition} for detail). 
%
For $\Pseq=(P_t)_{t\geq 1}$, 
$\pi\in (0,1]^V$ and a parameter $\epsilon>0$,
we define the separation time $\tsep^{(\pi)}(\Pseq,\epsilon)$ by
    \begin{align}
        \tsep^{(\pi)}(\Pseq,\epsilon)\defeq \min\left\{t\geq 0:\max_{s\geq 0, u,v\in V}\left(1-\frac{(P_{s+1}P_{s+2}\cdots P_{s+t})(u,v)}{\pi(v)}\right)\leq \epsilon\right\}
        \label{def:tsep}
    \end{align}
and let $\tsep^{(\pi)}(\mathcal{P})=\tsep^{(\pi)}(\Pseq,1/2)$.
If $P_t=P$ for all $t\geq 1$, the definition \cref{def:tsep} coincides with well-known definition of the separation time for static $P$ in the literature (see, e.g., Section 4.3 in \cite{AF02}).
Note that $\tsep^{(\pi)}(\mathcal{P})\leq\tmix^{(\infty,\pi)}(\mathcal{P})$ by definition. 


\begin{theorem}[Main result 2]
\label{thm:main_hitting_cover}
Let $\mathcal{P}=(P_t)_{t\geq 1}$ be a sequence of irreducible and reversible transition matrices.
Suppose that all $P_t$ have the same stationary distribution $\pi$.
Then, the following holds.
\begin{enumerate}[label=(\roman*)]
\item \label{lab:rev_hit}
$\thit^{(k)}(\Pseq)=O\left(\tsep^{(\pi)}(\Pseq)+\frac{\thitmax(\Pseq)}{k}\right)$ for any $k\geq 1$.
In particular, $\thit(\Pseq)=O\left(\thitmax(\Pseq)\right)$
if $P_t$ is lazy for all $t\geq 1$.
\item \label{lab:rev_cov}
$\tcov^{(k)}(\Pseq)=O\left(\tsep^{(\pi)}(\Pseq)+\frac{\thitmax(\Pseq)\log n}{k}\right)$ for any $k\geq 1$.
In particular, $\tcov(\Pseq)=O\left(\thitmax(\Pseq)\log n\right)$
if $P_t$ is lazy for all $t\geq 1$.
\end{enumerate}
\end{theorem}
\Cref{thm:main_hitting_cover} is the first result concerning multiple random walks on dynamic graphs.
For $k$-independent random walks according to a static $P$, the following bounds are known:
$\thit^{(k)}(P)=O\left(\tsep^{(\pi)}(\Pseq)+\frac{\thit(P)}{k}\right)$  (Theorem 8 in \cite{ER09}) and
$\tcov^{(k)}(P)=O\left(\tsep^{(\pi)}(\Pseq)+\frac{\thit(P)\log n}{k}\right)$ (Theorem 3.2 in~\cite{ES11}).
Hence, \cref{thm:main_hitting_cover} can be seen as a generalization
of these previous bounds.
Furthermore, for the case of $k=1$, \cref{thm:main_hitting_cover} improves various previous bounds of $\thit(\mathcal{P})$ and $\tcov(\mathcal{P})$.
See \cref{sec:example} for details.

\paragraph*{Meeting time and coalescing time.}
Our third result concerns the \emph{meeting time} $\tmeet(\Pseq)$ and the \emph{coalescing time} $\tcoal(\Pseq)$ of random walks on dynamic graphs.
Consider two independent
    random walks according to
    the same transition matrix sequence $\mathcal{P}$.
The meeting time $\tmeet(\mathcal{P})$
    is the expected time
    for the two walkers
    to meet starting from the worst initial positions.
In the coalescing random walk,
    we consider $|V|$ independent random walks
    according to $\mathcal{P}$
    starting from $|V|$ distinct initial positions.
Once two or more walkers gather at the same position,
    the walkers are merged into one walker.
The coalescing time $\tcoal(\mathcal{P})$
    is the expected time
    for the $|V|$ walkers
    to merge into one walker.
    By definition, we have $\tmeet(\Pseq)\leq \tcoal(\Pseq)$ in general.

For any irreducible, reversible, and lazy $P$, it is known that $\tcoal(P)=O(\thit(P))$ (Theorem 1.4 in \cite{OP19}).
Similarly to the hitting time,
we observe an exponential gap
    between static and dynamic settings.
\begin{proposition}
\label{prop:exponential_lower_meet}
There is a sequence of irreducible, reversible, and lazy transition matrices $\Pseq=(P_t)_{t\geq 1}$ satisfying $\tmeet(\Pseq)=2^{\Omega(n)}$.  
\end{proposition}
Note that \cref{prop:exponential_lower_meet} also gives an exponential lower bound of the coalescing time since $\tcoal(\Pseq)\geq \tmeet(\Pseq)$.
The random walk we consider in \cref{prop:exponential_lower_meet}
is the lazy simple random walk
on the graph sequence presented by
Olshevsky and Tsitsiklis~\cite{OT11},
which has an exponential meeting time.

The following main result
    presents further evidence
    that
    that a time-inhomogeneous coalescing walk with a common stationary distribution
    behaves almost identically to that on a static graph.
\begin{theorem}[Main result 3]
\label{thm:main_coalescing}
Let $\mathcal{P}=(P_t)_{t\geq 1}$ be a sequence of irreducible, reversible, and lazy transition matrices.
Suppose that all $P_t$ have the same stationary distribution $\pi$.
Then, $\tcoal(\Pseq)=O\left(\thitmax(\mathcal{P})\right)$.
\end{theorem}
\Cref{thm:main_coalescing} 
is the first result for the coalescing time on dynamic graphs. 
\Cref{thm:main_coalescing} 
generalizes
the aforementioned bound of $\tcoal(P)=O(\thit(P))$ for static $P$ (Theorem 1.4 in \cite{OP19}).
Furthermore, \cref{thm:main_coalescing} plays a key role to bound the consensus time of the \emph{pull voting} on dynamic graphs
(\cref{sec:application_to_pull_voting}).
\subsection{Examples}
\label{sec:example}
Our general results \cref{thm:main_mixing,thm:main_hitting_cover,thm:main_coalescing}
    yield new and improved bounds
    for mixing, hitting, cover, and coalescing times for various
    concrete random walks.
These consequences are summarized in  \cref{tbl:example}.
Throughout this section, unless otherwise stated, we assume that a graph $G$ is connected.
For a graph $G$ and a vertex $v\in V(G)$, let $N(G,v)$ be the set of neighboring vertices of $v$ (excluding $v$) and $\deg(G,v)=|N(G,v)|$ be the degree of $v$.

\begin{table}[t]
\centering
\begin{tabular}{|c|c|cc|cc|cc|}
\hline
\multicolumn{2}{|c|}{}                                                                                                                                          & \multicolumn{2}{c|}{$\tmix^{(\infty,\pi)}$}                                                & \multicolumn{2}{c|}{$\thit$}                                                                                                          & \multicolumn{2}{c|}{$\tcov$}                                                                                                                \\ \hline
\multicolumn{2}{|c|}{}                                                                                                                                          & $O(\trelmax\log \pi_{\min}^{-1})$     & Pro.$\ref{prop:mixing_SZ07}$                       & \cellcolor[HTML]{EFEFEF}                                & \cellcolor[HTML]{EFEFEF}                                                    & \cellcolor[HTML]{EFEFEF}                                      & \cellcolor[HTML]{EFEFEF}                                                    \\
\multicolumn{2}{|c|}{\multirow{-2}{*}{General}}                                                                                                                 & \cellcolor[HTML]{EFEFEF}$O(\thitmax)$ & \cellcolor[HTML]{EFEFEF}Th.$\ref{thm:main_mixing}$ & \multirow{-2}{*}{\cellcolor[HTML]{EFEFEF}$O(\thitmax)$} & \multirow{-2}{*}{\cellcolor[HTML]{EFEFEF}Th.$\ref{thm:main_hitting_cover}$} & \multirow{-2}{*}{\cellcolor[HTML]{EFEFEF}$O(\thitmax\log n)$} & \multirow{-2}{*}{\cellcolor[HTML]{EFEFEF}Th.$\ref{thm:main_hitting_cover}$} \\ \hline
                                                                             &                                                                                  &                                       &                                                    & $O(n^3\log n)$                                          & $\cite{SZ19}$                                                               & $O(n^3\log^2n)$                                               & $\cite{SZ19}$                                                               \\
                                                                             & \multirow{-2}{*}{\begin{tabular}[c]{@{}c@{}}Time hom.\\ deg. dist.\end{tabular}} & \multirow{-2}{*}{$O(n^3)$}            & \multirow{-2}{*}{$\cite{SZ19}$}                    & \cellcolor[HTML]{EFEFEF}$O(n^3)$                        & \cellcolor[HTML]{EFEFEF}Th.$\ref{thm:main_hitting_cover}$                   & \cellcolor[HTML]{EFEFEF}$O(n^3\log n)$                        & \cellcolor[HTML]{EFEFEF}Th.$\ref{thm:main_hitting_cover}$                   \\ \cline{2-8} 
\multirow{-3}{*}{\begin{tabular}[c]{@{}c@{}}LS\\ RW\end{tabular}}            & Regular                                                                          & $O(n^2)$                              & $\cite{SZ19}$                                      & $O(n^2)$                                                & $\cite{SZ19}$                                                               & $O(n^2\log n)$                                                & $\cite{SZ19}$                                                               \\ \hline
                                                                             &                                                                                  & $O(n^3\log n)$                        & $\cite{AKL08}$                                     & $O(n^3\log n)$                                          & $\cite{DR14}$                                                               &                                                               &                                                                             \\
\multirow{-2}{*}{\begin{tabular}[c]{@{}c@{}}$d_{\max}$\\ -lazy\end{tabular}} & \multirow{-2}{*}{Any graph}                                                      & \cellcolor[HTML]{EFEFEF}$O(n^3)$      & \cellcolor[HTML]{EFEFEF}Th.$\ref{thm:main_mixing}$ & \cellcolor[HTML]{EFEFEF}$O(n^3)$                        & \cellcolor[HTML]{EFEFEF}Th.$\ref{thm:main_hitting_cover}$                   & \multirow{-2}{*}{$O(n^3\log n)$}                              & \multirow{-2}{*}{$\cite{DR14}$}                                             \\ \hline
LMW                                                                          & Any graph                                                                        & \cellcolor[HTML]{EFEFEF}$O(n^2)$      & \cellcolor[HTML]{EFEFEF}Th.$\ref{thm:main_mixing}$ & \cellcolor[HTML]{EFEFEF}$O(n^2)$                        & \cellcolor[HTML]{EFEFEF}Th.$\ref{thm:main_hitting_cover}$                   & \cellcolor[HTML]{EFEFEF}$O(n^2\log n)$                        & \cellcolor[HTML]{EFEFEF}Th.$\ref{thm:main_hitting_cover}$                   \\ \hline
\end{tabular}
\caption{Bounds for specific examples that can be obtained from \cref{thm:main_mixing,thm:main_hitting_cover}.
The general results from \cref{prop:mixing_SZ07,thm:main_mixing,thm:main_hitting_cover}
    are given in General row.
The examples include bounds for a lazy simple random walk (LSRW) on graphs with a time-homogeneous degree distribution, regular graphs, a $d_{\max}$-lazy walk, and a lazy Metropolis walk (LMW).
Gray cells mean new or improved bounds.
}
\label{tbl:example}
\end{table}

\paragraph*{Lazy simple random walk.}
The transition matrix $\LSimP(G)$ of the lazy simple random walk\footnote{The laziness does not change the order of the hitting and cover times.
On the other hand, on connected bipartite graphs, the mixing and coalescing times of the simple random walk are unbounded, while these are bounded for the lazy simple random walk.
Hence, we assume the laziness in many cases.}
on a graph $G$ is defined by 
$\LSimP(G)(u,v)=\frac{1}{2\deg(G,u)}$ if $v\in N(G,u)$, 
$\LSimP(G)(u,u)=1/2$, 
and $\LSimP(G)(u,v)=0$ otherwise.
It is well known that $\thit(\LSimP(G))=O(n^3)$ for any $G$~\cite{Aleliunas79}, $\thit(\LSimP(G))=O(n^2)$ for any regular  $G$~\cite{KLNS89}, and $\thit(\LSimP(G))=O(n)$ for any regular expander\footnote{A graph is expander if $\trel(\LSimP(G))\leq C$ for some constnat $C>0$.} $G$~\cite{AF02}.

Let $\mathcal{G}=(G_t)_{t\geq 1}$ be a sequence of graphs with a time homogeneous degree distribution $(d(v))_{v\in V}$,  
i.e., $\deg(G_t,v)=d(v)$ holds for all $t\geq 1$ and $v\in V$.
Then, 
$\Pseq=(\LSimP(G_t))_{t\geq 1}$ has the common stationary distribution $\pi=\left(\frac{d(v)}{\sum_{u\in V}d(u)}\right)_{v\in V}$.
Hence, we can apply \cref{thm:main_mixing,thm:main_hitting_cover,thm:main_coalescing}.
In particular, \cref{thm:main_hitting_cover} implies  $\thit(\Pseq)=O(\thitmax(\Pseq))=O(n^3)$. 
This improves the $O(n^3\log n)$ bound (Main Result 1(3)) of \cite{SZ19}.
Another interesting example is the
sequence
$\mathcal{G}=(G_t)_{t\geq 1}$ of regular expander graphs.
Let $\mathcal{P}=(\LSimP(G_t))_{t\geq 1}$.
For $k=O(n)$, the cover time of $k$ independent lazy simple random walks
satisfies $\tcov^{(k)}(\Pseq)=O\bigl(\frac{n\log n}{k}\bigr)$ from \cref{prop:mixing_SZ07,thm:main_hitting_cover}.
This bound is tight since $\tcov^{(k)}(\LSimP(G))=\Omega\bigl(\frac{n\log n}{k}\bigr)$ holds for any (static) $G$ and $k=O(n\log n)$ \cite{RSS21}.

%
\paragraph{$d_{\max}$-lazy walk.} 
Let $d_{\max}=d_{\max}(G)\defeq \max_{v\in V(G)}\deg(v)$ denote the maximum degree of $G$.
The transition matrix $\MDW(G)$ of the $d_{\max}$-lazy walk on a graph $G$ is defined by
$\MDW(u,v)=\frac{1}{2d_{\max}}$ if $\{u,v\}\in E(G)$, 
$\MDW(u,u)=1-\frac{\deg(G,u)}{2d_{\max}}$, 
and $\MDW(G)(u,v)=0$ otherwise.
It is known that $\thit(\MDW(G))=O(n^3)$ holds for any $G$ (\cite{AKL18,DR14}).

Note that $\MDW(G)$ has the uniform stationary distribution for any $G$ since $\MDW(G)$ is symmetric. 
Hence, for any sequence of graphs $\mathcal{G}=(G_t)_{t\geq 1}$ and $\Pseq=(\MDW(G_t))_{t\geq 1}$, 
we can apply \cref{thm:main_mixing,thm:main_hitting_cover,thm:main_coalescing}. 
For example, $\thit(\Pseq)=O(\thitmax(\Pseq))=O(n^3)$ holds from \cref{thm:main_hitting_cover}.
This improves the previous $O(n^3\log n)$ bound  in \cite{DR14}.
%
\paragraph{Metropolis walk.}

We saw in the previous paragraph
    that the $d_{\max}$-lazy walk
    has a polynomial cover time
    for \emph{any} dynamic graph.
However, as mentioned in \cite{AKL18}, the $d_{\max}$-lazy random walk requires knowledge of the maximum degree, which is a \emph{global} information of $G_t$ at each $t\geq 1$.
Regarding this issue,
    we consider the lazy \emph{Metropolis walk}
    of Nonaka, Ono, Sadakane, and Yamashita~\cite{NOSY10}.
The transition matrix $\LMWP(G)$ of the lazy Metropolis walk on a graph $G$ is  defined by
\begin{align}
      \LMWP(G)(u,v)=\begin{cases}
        \frac{1}{2\max\{\deg(G,u),\deg(G,v)\}} & \text{if $\{u,v\}\in E(G)$},\\
        1-\sum_{w\in N(G,u)}\frac{1}{2\max\{\deg(G,u),\deg(G,w)\}} & \text{if $u=v$},\\
        0 & \text{otherwise}.
        \end{cases} \label{def:Metropolis_walk}
\end{align}
Note that
    a lazy Metropolis walk
    uses \emph{local} degree information around the walker.
For any $G$, Nonaka et al.~\cite{NOSY10} showed $\thit(\LMWP(G))=O(n^2)$.
    
Since 
the transition matrix of the lazy Metropolis walk 
is symmetric, the stationary distribution is uniform
        for any underlying graph.
Hence, we can apply \cref{thm:main_mixing,thm:main_hitting_cover,thm:main_coalescing}.
Interestingly, our bounds for Metropolis walks in \cref{tbl:example}
    are tight up to a constant factor:
On the (static) cycle graph, the lazy Metropolis walk has $\Omega(n^2)$ mixing time and  $\Omega(n^2)$ hitting time~\cite{AF02}.
On the glitter star graph of \cite{NOSY10}, the lazy Metropolis walk has an $\Omega(n^2\log n)$ cover time.

\paragraph*{Metropolis walks on edge-Markovian graphs.}
Our results
concern the sequence of
connected graphs.
Indeed, it is not difficult to see that \cref{thm:main_mixing,thm:main_hitting_cover,thm:main_coalescing} also hold if
$G_t$ is connected at least once in every $C$ steps for some positive constant $C$.
This setting was already studied in \cite{DR14}.
With some additional arguments, we study random walks
on the \emph{edge-Markovian graph} $(G_t)_{t\geq 1}$ defined as follows:
Let $p,q\in[0,1]$ be parameters and
$G_1$ be an arbitrary fixed graph.
The graph $G_{t+1}$ is obtained by adding each $e\in \binom{V}{2}\setminus E(G_{t})$ independently with probability $p$ and removing each $e\in E(G_t)$ independently with probability $q$.

The model of edge-Markovian graph was introduced by Clementi, Macci, Monti, Pasquale, and Silvestri~\cite{CMMPS10}
as a wide generalization of time-independent dynamic random graphs.
Since then,
    several properties
    including the
    flooding~\cite{CMMPS10,BCF11},
    rumor spreading~\cite{CCDFPS16}, and mixing time~\cite{CSZ20}
    have been investigated.
%
In this paper, we focus on the Metropolis walk on the edge-Markovian graph $(G_t)_{t\geq 1}$
and obtain the following result.
See \cref{sec:edge_Markovian} for the proof.

\begin{theorem}
\label{thm:LMRW_EM}
Let $c\geq 1$ be arbitrary.
Consider $\mathcal{G}(n,p,q)=(G_t)_{\geq 0}$ the edge-Markovian graph satisfying $\frac{p}{p+q}\geq 32(c+1)\frac{\log n}{n}$ and $0<p+q\leq 1$.
Let $\Pseq=(\LMWP(G_t))_{t\geq 1}$.
Then, for any $k\geq 1$, $\mathcal{G}(n,p,q)$ satisfies the following with probability $1-n^{-\Omega(1)}$:
$\thit^{(k)}(\Pseq)=O\left(\frac{\max\{1,q/p\}}{p+q}+\log n+ \frac{n}{k}\right)$, 
$\tcov^{(k)}(\Pseq)=O\left(\frac{\max\{1,q/p\}}{p+q}+\log n+\frac{n\log n}{k}\right)$ 
, and $\tcoal(\Pseq)=O\left(\frac{\max\{1,q/p\}}{p+q}+n\right)$.
\end{theorem}

%




\subsection{Pull voting on dynamic graphs} \label{sec:application_to_pull_voting}
In the \emph{(weighted) pull voting},
we consider an $n$-vertex graph $G=(V,E)$
where each vertex $v\in V$ holds
an opinion $\sigma_v\in\Sigma$
for a finite set $\Sigma\subseteq\{0,\dots,n-1\}$
of possible opinions.
Let $P\in[0,1]^{V\times V}$ be a
    transition matrix.
At every discrete time step,
each vertex $u$ chooses a random neighbor
according to the distribution $P(u,\cdot)$
and then
updates
its opinion with the neighbor's opinion.
The aim of the protocol is to
reach consensus in which
every vertex supports the same opinion.
The \emph{consensus time} $\tau_{\mathrm{cons}}$
is the time of the process to reach
consensus.

We consider the \emph{pull voting on dynamic graphs}.
Specifically, let
$\Pseq=(P_t)_{t\geq 1}$ be a transition matrix sequence.
In the \emph{pull voting according to $\Pseq$},
at the $t$-th round,
vertices perform the one-round pull voting according to $P_t$.
We denote by $\taucons(\Pseq)$ the consensus time of the
pull voting according to $\Pseq$
and consider $\tcons(\Pseq)\defeq\E[\taucons(\Pseq)]$.
Combining \cref{thm:main_coalescing} and the idea of well-known \emph{duality}
    between the coalescing random walk and pull voting~\cite{HP01}
    with some additional argument,
    we prove the following.

\begin{theorem}[Consensus time] \label{thm:consensus_time}
Let $\mathcal{P}=(P_t)_{t\geq 1}$ be a sequence of irreducible, lazy, and reversible transition matrices.
Suppose that all $P_t$ have the same stationary distribution $\pi$.
Then, $\tcons(\Pseq)\leq O(\thitmax(\Pseq))$.
\end{theorem}

On the other hand,
    the consensus time
    can be exponential in general.


\begin{proposition}\label{prop:tcons_lower_bound}
There is a sequence $(G_t)_{t\geq 1}$ of $n$-vertex connected graphs such that
    $\tcons(\Pseq) = 2^{\Omega(n)}$
    for $\Pseq=(\LSimP(G_t))_{t\geq 1}$.
\end{proposition}

Another important question concerning pull voting is the probability that the process finally agrees with a specific opinion $\sigma\in\Sigma$.
Using the \emph{voting martingale} argument (e.g., \cite{CR16}), we obtain the following result.
\begin{proposition}[Winning probability] \label{prop:winning_probability_dynamic}
Let $\mathcal{P}=(P_t)_{t\geq 1}$ be a sequence of irreducible, lazy, and reversible transition matrices.
Suppose that all $P_t$ have the same stationary distribution $\pi$.
Consider the pull voting over opinion set $\Sigma$ according to $\Pseq$.
Then, the process finally agrees with $\sigma\in\Sigma$ with probability
$\sum_{v\in V_\sigma} \pi(v)$,
where $V_\sigma$ is the set
    of vertices initially holding
    opinion $\sigma$.
\end{proposition}

\subsection{Proof overview}
\label{sec:proof_outline}
We overview the proof of the main theorems (\cref{thm:main_mixing,thm:main_hitting_cover,thm:main_coalescing}).
%
Throughout this section, we assume that 
$\Pseq=(P_t)_{t\geq 1}$ is a sequence of irreducible and reversible transition matrices in which all $P_t$ have the same stationary distribution $\pi\in (0,1]^V$.
For $f,g\in\Real^V$, define the inner product $\langle\cdot,\cdot\rangle_\pi$ as $\langle f,g\rangle_\pi =\sum_{v\in V}f(v)g(v)\pi(v)$ and the induced norm $\|f\|_{2,\pi}=\sqrt{\sum_{v\in V}f(v)^2\pi(v)}$.
Define $\frac{f}{\pi}\in \Real^V$ by $\left(\frac{f}{\pi}\right)(v)=\frac{f(v)}{\pi(v)}$.
Let $\mathbbm{1}$ denote the $|V|$-dimensional all-one vector. 

For a vertex $w\in V$, let $D_w\in \{0,1\}^{V\times V}$ be the diagonal matrix defined by $D_w(v,v)=1$ if $v\neq w$ and $D_w(w,w)=0$.
We are interested in
    the substochastic matrix $D_wPD_w$.
Note that $D_wPD_w$ is the matrix obtained by replacing elements of $P(w,\cdot)$ and $P(\cdot,w)$ with $0$.
It is known that $\rho(D_wPD_w)\leq 1-\frac{1}{\thit(P)}$,
where $\rho(M)$ is the spectral
    radius of $M$
    (see \cref{lem:hiteigen} or Section 3.6.5 in \cite{AF02}).
    

%
\paragraph*{Mixing time (\cref{sec:mixing_time}).}
Let $\mu_0\in [0,1]^V$ be an initial distribution and $\mu_T=\mu_0 \prod_{i=1}^T P_i$ be the distribution of $X_T$.
Let $d^{(2,\pi)}(\mu_T)\defeq \left\|\frac{\mu_T}{\pi}-\mathbbm{1}\right\|_{2,\pi}$, 
i.e., the $\ell^2$-distance between $\mu_T/\pi$ and $\mathbbm{1}$.
It is known that $\max_{v\in V}\left|\frac{\mu_T(v)}{\pi(v)}-1\right|$ can be bounded
    in terms of $d^{(2,\pi)}(\cdot)$
    (see, e.g., \cite{LP17,SZ07} or \cref{eq:l2_to_unif}).
Henceforth, we focus on bounding $d^{(2,\pi)}(\mu_T)$.

The main part of the proof of \cref{thm:main_mixing} is to prove the following inequality:
For any $t\geq 0$ and \emph{lazy} $\Pseq=(P_t)_{t\geq 1}$, $d^{(2,\pi)}(\mu_{t+1})^2\leq d^{(2,\pi)}(\mu_t)^2\left(1-\frac{d^{(2,\pi)}(\mu_t)^2}{\thitmax(\Pseq)}\right)$ holds.
Applying this inequality repeatedly, we obtain $\tmix^{(\infty,\pi)}(\Pseq)=O(\thitmax(\Pseq))$ 
(see \cref{lem:mix_teq_SZ} or \cite{SZ19} for detail).

The aforementioned inequality comes from a variant of Mihail's identity (\cref{lem:Mihail}): For any reversible and \emph{lazy} $P$ and any probability vector $\mu$, $d^{(2,\pi)}(\mu P)^2\leq d^{(2,\pi)}(\mu)^2-\mathcal{E}_{P,\pi}(\mu/\pi)$ holds.
Here, $\mathcal{E}_{P,\pi}(f)=\frac{1}{2}\sum_{u,v\in V}\pi(u)P(u,v)(f(u)-f(v))^2=\langle f,f \rangle_\pi-\langle f,Pf \rangle_\pi$ is the Dirichlet form of $P$ and $\pi$.
In~\cite{SZ19}, authors consider a lazy simple random walk and give lower bounds of $\mathcal{E}_{\LSimP,\pi}(\mu/\pi)$ in terms of some graph parameters, e.g., $\mathcal{E}_{\LSimP,\pi}(\mu/\pi)\geq d^{(2,\pi)}(\mu)^4/n^2$ for regular graphs.
This lower bound means $d^{(2,\pi)}(\mu_{t+1})^2\leq d^{(2,\pi)}(\mu_t)^2\left(1-\frac{d^{(2,\pi)}(\mu_t)^2}{n^2}\right)$ holds.
%
Our technical contribution is to generalize the previous lower bound of $\mathcal{E}_{P,\pi}(\mu/\pi)$ for any reversible and lazy $P$ in terms of the hitting time (\cref{lem:dirichlethit}): $\mathcal{E}_{P,\pi}(\mu/\pi)\geq d^{(2,\pi)}(\mu)^4/\thit(P)$ holds.
This implies the desired inequality, $d^{(2,\pi)}(\mu_{t+1})^2\leq d^{(2,\pi)}(\mu_t)^2\left(1-\frac{d^{(2,\pi)}(\mu_t)^2}{\thitmax(\Pseq)}\right)$.

The proof of the key lemma (\cref{lem:dirichlethit}) 
    consists of four steps.
First, observe that $\mathcal{E}_{P,\pi}(f)=\mathcal{E}_{P,\pi}(-f)$ and  $\mathcal{E}_{P,\pi}(f)=\mathcal{E}_{P,\pi}(f+c\mathbbm{1})$ for any $f\in\mathbb{R}^V$ and $c\in\mathbb{R}$.
Hence, 
$\mathcal{E}_{P,\pi}(f)=\mathcal{E}_{P,\pi}(f_{\max}\mathbbm{1}-f)$ holds for $f_{\max}\defeq \max_{v\in V}f(v)$.
Second, 
let $g\defeq f_{\max}\mathbbm{1}-f$.
Note that $g(w)=0$ for $w\in V$ satisfying $f(w)=f_{\max}$.
Therefore, for any $u,v\in V$, we have
    $g(v)P(v,u)g(u)=g(v)(D_wPD_w)g(u)$
    and thus $\langle g,Pg \rangle_\pi = \langle g, D_wPD_w g \rangle_\pi$ holds.
Hence, we have
$\mathcal{E}_{P,\pi}(g)=\langle g,g \rangle_\pi - \langle g,Pg\rangle_\pi = \mathcal{E}_{D_wPD_w,\pi}(g)$.
Third, from the known inequality 
    $\langle g,D_wPD_wg\rangle_\pi \leq \rho(D_wPD_w) \|g\|_{2,\pi}^2$ for the spectral radius (\cref{lem:MatrixEigen})
    and $\rho(D_wPD_w) \leq 1-\frac{1}{\thit(P)}$ (\cref{lem:hiteigen}),
we have $\mathcal{E}_{P,\pi}(g)=\mathcal{E}_{D_{w}PD_{w},\pi}(g)\geq \frac{\|g\|_{2,\pi}^2}{\thit(P)}$.
Finally, from a carefully calculation, it is not difficult to see that $\|g\|_{2,\pi}^2\geq d^{(2,\pi)}(\mu)^4$ for $f=\mu/\pi$.

\paragraph{Hitting and cover times (\cref{sec:hitting_and_cover_times}).}
Let $(X_t)_{t\geq 0}$ be the random walk according to $\Pseq=(P_t)_{t\geq 1}$.
To obtain an upper bound of the hitting time, it suffices to bound the probability $\Pr\left[\land_{t=0}^T\{X_t\neq w\}\right]$ for any fixed vertex $w$.
Sauerwald and Zanetti~\cite{SZ19}
    used the conditional expectation approach to bound this probability.
Instead of this strategy, we use the following key lemma.
\begin{lemma}
\label{lem:HTL}
Suppose that $X_0$ is sampled from $\pi$.
Then, for any vertex $w\in V$ and $T\geq 0$,
\begin{align*}
    \Pr\left[\bigwedge_{t=0}^T\{X_t\neq w\}\right]
    \leq \prod_{t=1}^T\left(1-\frac{1}{\thit(P_t)}\right)
    \leq \exp\left(-\frac{T}{\thitmax(\Pseq)}\right).
\end{align*}
\end{lemma}

If the walker starts according to the stationary distribution, \cref{lem:HTL} immediately gives the bounds of the hitting and cover times:
Since the event $\land_{t=0}^T\{X_t\neq w\}$ means that the walk does not hit $w$ until time step $T$, the expected hitting time is upper bounded by $\sum_{T=0}^\infty\left(1-\frac{1}{\thitmax(\Pseq)}\right)^T=\thitmax(\Pseq)$.
From the union bound, the probability that the cover time is larger than $2\thitmax(\Pseq)\log n$ is upper bounded by $n\left(1-\frac{1}{\thitmax(\Pseq)}\right)^{2\thitmax(\Pseq)\log n}\leq 1/n$.
This argument can be
    easily extended to the case of
    $k$ independent random walks starting from positions according to the stationary distribution; the
    expected hitting time is bounded
    by $\sum_{T=0}^\infty\left(1-\frac{1}{\thitmax(\Pseq)}\right)^{kT}=O(\thitmax(\Pseq)/k)$.


\Cref{lem:HTL} and the bounds of the separation distance (\cref{thm:main_mixing}) enable us to obtain upper bounds of the hitting and cover times of $k$ walkers from the \emph{worst} initial positions.
Here is the proof sketch.
Let $(X_t(1),\dots,X_t(k))_{t\geq 0}$ be $k$ independent walks starting from the worst initial positions.
By definition of the separation distance \cref{def:tsep}, the probability that the walker $i\in [k]$ is on $u\in V$ at time $\tsep^{(\pi)}(\Pseq)$ is $P_{[1,\tsep^{(\pi)}(\Pseq)]}(X_0(i),u)\geq \frac{1}{2}\pi(u)$.
Therefore, in expectation, half of the walkers are distributed according to $\pi$ after $\tsep^{(\pi)}(\Pseq)$ steps.
From \cref{lem:HTL}, the probability that all of such walkers do not hit a specific vertex $w$ within $T$ steps is at most roughly
    $(1-1/\thitmax(\Pseq))^{Tk/2}$.
Taking $T=O(\thitmax(\Pseq)/k)$, this probability can be bounded by some constant probability.
This implies $\thit^{(k)}(\Pseq)=O(\tsep^{(\pi)}(\Pseq)+\thitmax(\Pseq)/k)$ (see \cref{lem:key_lemma} for detail).
The proof for the cover time proceeds in a similar way: From the union bound over $w\in V$, the probability that there is a vertex $w$ such that all the $\pi$-distributed walkers do not hit $w$ within $T$ steps is at most roughly
    $n(1-1/\thitmax(\Pseq))^{Tk/2}$.


The proof of \cref{lem:HTL} goes as follows.
For a fixed $w\in V$,
    consider the sequence of substochastic matrices
    $(D_wP_tD_w)_{t\geq 1}$.
By definition of $D_w$ and the Cauchy--Schwarz inequality, we have
$\Pr\left[\bigwedge_{t=0}^T\{X_t\neq w\}\right]=\sum_{x,y\in V}\pi(x)\left(\prod_{t=1}^T(D_wP_tD_w)\right)(x,y)\leq \left\|\left(\prod_{t=1}^T(D_wP_tD_w)\right)\mathbbm{1}\right\|_{2,\pi}$. 
Furthermore, from the reversibility of $P_t$, all $D_wP_tD_w$ are reversible. 
Thus we can apply a variant
of the Courant--Fischer theorem for the $\pi$-inner product (\cref{lem:MatrixEigen}) repeatedly and obtain $\left\|\left(\prod_{t=1}^T(D_wP_tD_w)\right)\mathbbm{1}\right\|_{2,\pi}\leq \prod_{t=1}^T\rho(D_wP_tD_w)$. 
Since $\rho(D_wPD_w)\leq 1-1/\thit(P)$ (\cref{lem:hiteigen}),
 we obtain \cref{lem:HTL}.

\paragraph*{Meeting and coalescing times (\cref{sec:coalescing_time}).}
To give an upper bound of the coalescing time, we recall the powerful \emph{Meeting Time Lemma} given by Oliveira~\cite{Oliveira12} (for continuous-time walks) and Oliveira and Peres~\cite{OP19} (for discrete-time lazy  walks).
These are originally results for
time-homogeneous random walks.
Our key observation is that Meeting Time Lemma indeed holds for
    time-inhomogeneous random walks with the same stationary distribution.
Suppose that $\Pseq$ is lazy (Note that \cref{lem:HTL} does not need the laziness assumption).
\begin{lemma}
\label{lem:MTL}
Suppose that $X_0$ is sampled from $\pi$.
Then, for any sequence of vertices $(w_t)_{t\geq 0}$,
\begin{align*}
    \Pr\left[\bigwedge_{t=0}^T \{X_t\neq w_t\}\right] \leq
    \prod_{t=1}^T \left(1-\frac{1}{\thit(P_t)}\right)
    \leq \exp\left(-\frac{T}{\thitmax(\Pseq)}\right).
\end{align*}
\end{lemma}
Suppose that the initial positions of two independent random walks $(X_t(1))_{t\geq 0}$ and $(X_t(2))_{t\geq 0}$ are according to the stationary distribution. 
\Cref{lem:MTL} gives an upper bound of the probability that two walkers do not meet in the dynamic setting: $\Pr\left[\land_{t=0}^T\{X_t(1)\neq X_t(2)\}\right]\leq \exp\left(-\frac{T}{\thitmax(\Pseq)}\right)$, i.e., the expected meeting time is upper bounded by $\thitmax(\Pseq)$.
Note that $\Pr\left[\land_{t=0}^T\{X_t(1)\neq X_t(2)\}\right]\leq \Pr\left[\land_{t=0}^T\{X_t(1)\neq w_t)\}\right]$ holds for some sequence of vertices $(w_t)_{t\geq 0}$.

The proof of our $O(\thitmax(\Pseq))$ coalescing time bound essentially consists of three parts.
In the first part, we bound
the probability that
coalescing time is larger than $C\thitmax(\Pseq)$ for a suitable constant $C>0$ 
in terms of the probability of a suitable event of independent multiple random walks (\cref{lem:coal_mult}).
This part proceeds similarly to the argument of static setting in \cite{OP19} by coupling arguments.
In the second part, similar to the arguments of multiple random walks, we show that the positions of an appropriate group of walkers at $t=\tsep^{(\pi)}(\Pseq)$ are according to the stationary distribution with constant probability.
In the third part,
we apply \cref{lem:MTL} to such group of walkers (\cref{lem:coalcons}). Since \cref{lem:MTL} works in the dynamic setting, we can complete this part in a similar way to the static setting \cite{OP19}. 

In the proof of \cref{lem:MTL}, the laziness assumption of walkers is essential as well as that of~\cite{OP19}. 
Since $P$ is positive semidefinite by the laziness,
we rewrite
$D_xPD_y=D_x\sqrt{P}\sqrt{P}D_y$, where $\sqrt{P}$ is the square root of $P$. 
Furthermore, the reversibility of $P$ implies that $D_x\sqrt{P}$ is the adjoint of $\sqrt{P}D_x$ (that is, $\pi(u)(D_x\sqrt{P})(u,v)=\pi(v)(\sqrt{P}D_x)(v,u)$ for any $u,v\in V$).
Hence, we have $\|D_xPD_yf\|_{2,\pi}\leq \sqrt{\rho(D_xPD_x)\rho(D_yPD_y)}\|f\|_{2,\pi}\leq \left(1-\frac{1}{\thit(P)}\right)\|f\|_{2,\pi}$ for any vector $f$ (\cref{lem:normmeet,lem:hiteigen}). 
Combining this inequality and similar arguments in the proof of \cref{lem:HTL}, we obtain $\Pr\left[\bigwedge_{t=0}^T\{X_t\neq w_t\}\right]
\leq \left\|\left(\prod_{t=1}^T(D_{w_{t-1}}P_tD_{w_t})\right)\mathbbm{1}\right\|_{2,\pi}\leq \prod_{t=1}^T\left(1-\frac{1}{\thit(P_t)}\right)$.


\subsection{Related work} \label{sec:previous_work}

\paragraph*{Random walk on a static graph.}
Consider a simple random walk on a connected graph $G$ of $n$ vertices and $m$ edges.
Aleliunas, Karp, Lipton, Lov{\'a}sz, and Rackoff~\cite{Aleliunas79} showed that 
the cover time is at most
$2m(n-1)$. 
Kahn, Linial, Nisan, and Saks~\cite{KLNS89} showed that the cover time is at most  $16mn/d_{\min}$,
while
the hitting time is at least $(1/2)m/d_{\min}$ (see Corollary 3.3 of Lov{\'a}sz~\cite{Lovasz93}).
%
Brightwell and Winkler~\cite{BW90} presented the lollipop graph on which the hitting time is approximately $(4/27)n^3$ as $n$ increases,
while
Feige~\cite{Feige95up} gave proved that the cover time is $(4/27)n^3+O(n^{5/2})$
for any graph.
In addition to the trivial relation of $\thit\leq \tcov$, it is known that $\tcov \leq \thit \log n$ holds for any $G$ (see Matthews~\cite{Matthews88}).

The cover time $\tcov^{(k)}$ of $k$ independent simple random walks has been investigated in \cite{BKRU94,AAKKLT11,ES11,ER09,RSS21}.
If $k$ walkers start from the stationary distribution, Broder, Karlin, Raghavan, and Upfal \cite{BKRU94} showed that the cover time is at most $O\bigl(\left(\frac{m}{k}\right)^2\log ^3n\bigr)$.
Very recently, Rivera, Sauerwald, and Sylvester~\cite{RSS21} proved an improved bound of $O\bigl(\bigl(\frac{m}{kd_{\min}}\bigr)^2\log^2n\bigr)$.
From the worst initial positions of $k$ walkers, 
Els{\"a}sser and Sauerwald \cite{ES11} showed that $\tcov^{(k)}=O\left(\tmix+\frac{\thit \log n}{k} \right)$ for $k\leq n$.

It is known that local degree information provides surprising power with random walks.
For example, 
the \emph{$\beta$-random walk} proposed by Ikeda, Kubo, Okumoto, and Yamashita~\cite{Ikeda03,Ikeda09} and the Metropolis walk proposed by Nonaka, Ono, Sadakane, and Yamashita~(\cite{NOSY10}, the definition is \cref{def:Metropolis_walk}) have 
the $O(n^2)$ hitting time and $O(n^2\log n)$ cover time for any $G$.
These bounds improve the worst-case $\Omega(n^3)$ hitting time of the simple random walk (on the lollipop graph). 
Recently,  
David and Feige~\cite{DF18} showed the $O(n^2)$ cover time for the \emph{minimum-degree random walk} proposed by Abdullah, Cooper, and Draief~\cite{ACD15}. 
This is best possible since any random walk on the path has $\Omega(n^2)$ cover time~\cite{Ikeda09}.
It is easy to see that
    the $\beta$-random walk and the minimum-degree random walk
    have exponential hitting times on the Sisyphus wheel.
The meeting and coalescing times have been well investigated in the context of distributed computation such as leader election and consensus protocols \cite{HP01}.
Consider the simple random walk on a connected and nonbipartite $G$.
Tetali and Winkler showed that the meeting time is at most $(16/27+o(1))n^3$. 
Hassin and Peleg~\cite{HP01} showed $\tcoal\leq \tmeet\log n$, while $\tmeet\leq \tcoal$ is trivial.
Recent works on the meeting and coalescing times consider the lazy simple random walk on a connected graph $G$ \cite{CEOR13,BGKM16,KMS19,OP19}.
For example, Kanade, Mallmann-Trenn, and Sauerwald~\cite{KMS19} showed $\tcoal=O\left(\tmeet\left(1+\sqrt{\frac{\tmix}{\tmeet}}\log n\right)\right)$.
Oliveira and Peres~\cite{OP19} proved $\tcoal=O(\thit)$.

\paragraph*{Random walk on dynamic graphs.}
Avin, Kouck{\'y} and Lotker~\cite{AKL18} presented
the Sisyphus wheel on which the hitting time is $2^{\Omega(n)}$ for the lazy simple random walk (\cref{fig:sisyphus}).
To avoid the issue of the exponential hitting time, Avin et al.~\cite{AKL18} considered the \emph{$d_{\max}$-lazy random walk} on $G$ 
and showed that the cover time of this random walk on any sequence of connected graphs is $O(n^5\log^2n)$. 
They also showed that the mixing time of the walk is $O(n^3\log n)$. 
Denysyuk and Rodrigues \cite{DR14} improved the bound of the cover time to $O(n^3\log n)$.
Sauerwald and Zanetti~\cite{SZ19}
    considered a lazy simple random
    walk on a sequence $(G_t)_{t\geq 1}$ of graphs such that all $G_t$ have a common degree distribution.
 They showed that the $O(n/\pi_{\min})$ mixing time and $O((n\log n)/\pi_{\min})$ hitting time. 
 Moreover, if all $G_t$ are $d$-regular, then the hitting time is $O(n^2)$.
 This bound matches that of static regular graphs~\cite{KLNS89}.

The study of time inhomogeneous Markov chains has applications in a wide range of fields, including consensus algorithm~\cite{OT11} and cryptography~\cite{MPS04}.
In an early work, 
Griffeath~\cite{Gri75} studied the ergodic theorem of time inhomogeneous Markov chains.
Saloff-Conste and Z\'{u}\~{n}iga~\cite{SZ09,SZ11}
considered the merging time for time inhomogeneous Markov chains
    in terms of the stability of stationary distributions.
In another paper~\cite{SZ07}, they obtained an upper bound of the mixing time for time inhomogeneous Markov chains under the assumption that the chains are irreducible and have a common stationary distribution.

Cai, Sauerwald, and Zaneti~\cite{CSZ20} considered the lazy simple random walk on a sequence of edge-Markovian random graphs.
They introduced the notion of mixing time on this sequence (note that the stationary distribution changes over time) and obtained several mixing time bounds. 
Lamprou, Martin, and Spirakis~\cite{LMS18} studied the cover time of the simple random walk on a variant of edge-Markovian random graphs.


\paragraph*{Pull voting.}
The pull voting according to $\LSimP(G)$ for a static graph $G$ has been intensively studied~\cite{HP01,CEOR13}
    in the literature of distributed computing and stochastic process.
It is widely known that
    the consensus time and the coalescing time
    are equal.
Therefore, bounds for coalescing time
    yields bounds for consensus time.
In particular, the consensus time on any (static) connected nonbipartite graphs
    is $O(n^3\log n)$ from Hassin and Peleg~\cite{HP01}.

Berenbrink, Giakkoupis, Kermarrec, and Mallmann-trenn~\cite{BGKM16} studied
the pull voting according to $(\LSimP(G_t))_{t\geq 1}$ for a sequence $(G_t)_{t\geq 1}$ of graphs constructed by an \emph{adaptive adversary}.
That is, for every $t$, the graph $G_t$
    can depend on the history of opinion configurations.
Under the assumption that all $G_t$ must have the same degree distribution,
they obtained an upper bound of $\taucons$ in
terms of the conductance of $G_t$
for the binary opinion setting (i.e., $|\Sigma|=2$).

%


\section{Notations and definitions}
\label{sec:mixing_time_definition}
This section defines the mixing, hitting, cover, meeting, and coalescing times formally.
For $b\geq a \geq 1$ and a sequence $(P_t)_{t\geq 1}$ of transition matrices, let $P_{[a,b]}\defeq P_aP_{a+1}\cdots P_b$.
For $\pi\in[0,1]^V$, let $\pi_{\min}\defeq \min_{v\in V}\pi(v)$.

For $p\geq 1$ and probability vectors $\mu\in [0,1]^V$ and $\pi\in (0,1]^V$, let
\begin{align*}
    d^{(p,\pi)}(\mu)&\defeq 
    \begin{cases}
    \left\|\frac{\mu}{\pi}-\mathbbm{1}\right\|_{p,\pi}=\left(\sum_{v\in V}\pi(v)\left|\frac{\mu(v)}{\pi(v)}-1\right|^p\right)^{1/p} &\text{if }1\leq p<\infty, \\
    \max_{v\in V}\left|\frac{\mu(v)}{\pi(v)}-1\right| & \text{if }p=\infty
    \end{cases}
\end{align*}
be the $\ell^p$-distance between $\mu/\pi$ and $\mathbbm{1}$.
It is known that $d^{(p,\pi)}(\mu)\leq d^{(p+1,\pi)}(\mu)$ holds for any $p\geq 1$ (see, e.g., Section 4.7 in \cite{LP17}).
For example, $\sum_{v\in V}|\mu(v)-\pi(v)|=d^{(1,\pi)}(\mu)\leq d^{(2,\pi)}(\mu)\leq d^{(\infty,\pi)}(\mu)$.
For $\Pseq=(P_t)_{t\geq 1}$, a probability vector $\pi\in (0,1]^V$, and $\epsilon>0$, we define the $\ell^p$-mixing time as
\begin{align*}
\tmix^{(p,\pi)}(\Pseq,\epsilon)\defeq \min\left\{t\geq 0: \max_{s\geq 1,v\in V}d^{(p,\pi)}\left(P_{[s+1,s+t]}(v,\cdotp)\right)\leq \epsilon\right\}.
\end{align*}
Write $\tmix^{(p,\pi)}(\Pseq)\defeq \tmix^{(p,\pi)}(\Pseq,1/2)$.
Consider $k$ independent random walks $(X_t(1))_{t\geq 0}, \ldots, (X_t(k))_{t\geq 0}$, where each walk is according to $\Pseq=(P_t)_{t\geq 1}$.
Let $\tauhit^{(k)}(\Pseq,w)$ (for $w\in V$) and $\taucov^{(k)}(\Pseq)$ be the
random variables denoting hitting and cover times of the $k$ random walks, respectively.
Formally,
\begin{align}
    \tauhit^{(k)}(\Pseq,w)&=\inf\left\{t\geq 0:\bigcup_{i\in[k],0\leq s\leq t} \{X_s(i)\} \ni w \right\}, \label{eq:hitting_time_def}\\
    \taucov^{(k)}(\Pseq) &= \inf\left\{t\geq 0:\bigcup_{i\in [k],0\leq s\leq t}\{X_s(i)\}=V\right\}. \label{eq:cover_time_def}
\end{align}
Let $\thit^{(k)}(\Pseq)\defeq \max_{x\in V^k,w\in V}\E\left[\tauhit^{(k)}(\Pseq,w)\middle| X_0=x\right]$ for the expected hitting time of $k$ random walks.
Here, $X_t=(X_t(1),\ldots,X_t(k))\in V^k$ is a vector-valued random variable. 
Similarly, the expected cover time of $k$ random walks is defined by $\tcov^{(k)}(\Pseq)\defeq \max_{x\in V^k}\E\left[\taucov^{(k)}(\Pseq)\middle| X_0=x\right]$.
In particular, let $\thit(\Pseq) = \thit^{(1)}(\Pseq)$ and $\tcov(\Pseq) = \tcov^{(1)}(\Pseq)$.
%
%
%

Let $(X_t(1))_{t\geq 0}$ and $(X_t(2))_{t\geq 0}$ be two independent random walks, where each walker is according to $\Pseq=(P_t)_{t\geq 1}$.
Write $X_t=(X_t(1),X_t(2))\in V^2$.
Then, let $\taumeet(\mathcal{P})\defeq \min\{t\geq 0:X_t(1)=X_t(2)\}$ and define the meeting time of $\Pseq$ as $\tmeet(\Pseq)\defeq \max_{x\in V^2}\E\left[\taumeet(\Pseq)\middle|X_0=x\right]$.

Let $(\Co_t(1))_{t\geq 0}, (\Co_t(2))_{t\geq 0}, \ldots, (\Co_t(n))_{t\geq 0}$ denote the coalescing random walks according to $\Pseq=(P_t)_{t\geq 1}$.
In the coalescing random walks, once two or more walkers meet at the same vertex, they merge into one walker.
Formally, from a given initial state $\Co_{0}=(\Co_{0}(1),\ldots,\Co_{0}(n))\in V^n$, we inductively determine $\Co_t(a)$ for each $t\geq 1$ and $a\in [n]$, as follows.
Suppose that $\Co_{t-1}=(\Co_{t-1}(1),\ldots,\Co_{t-1}(n))$ and $\Co_t(1), \ldots, \Co_t(a-1)$ are determined.
If there is some $b<a$ such that $\Co_{t-1}(a)=\Co_{t-1}(b)$, let $\Co_t(a)\defeq\Co_t(b)$. 
Otherwise, $\Co_t(a)$ is determined by the random walk according to $P_{t}$, i.e., $\Pr[\Co_{t}(a)=v|\Co_{t-1}(a)=u]=P_t(u,v)$ for $u,v\in V$.
%
For $x=(x_1,x_2,\ldots,x_n)\in  V^n$, let $S(x)\defeq \bigcup_{i=1}^n\{x_i\}$ (e.g., $S(x)=\{a,b\}$ for $x=(a,a,b)$).
Then, let $\taucoal(\Pseq)=\min\{t\geq 0: |S(\Co_t)|=1\}$
and define the coalescing time of $\Pseq$ as $\tcoal(\Pseq)\defeq \max_{x\in V^n}\E[\taucoal(\Pseq)|\Co_0=x]$.

\section{Mixing time}
\label{sec:mixing_time}
In this section, we show the following lemma, which immediately yields \cref{thm:main_mixing}.
\begin{lemma}
\label{lem:uniform_main}
Let $\mathcal{P}=(P_t)_{t\geq 1}$ be a sequence of irreducible, reversible, and lazy transition matrices.
Suppose that all $P_t$ have the same stationary distribution $\pi$.
Then, for any $u,v\in V$ and any $0<\epsilon<1$,  
$\left|\frac{P_{[1,2T]}(u,v)}{\pi(v)}-1\right|\leq \epsilon^2$ holds if
$T\geq  \frac{\mathrm{e}^2}{\mathrm{e}-1}\thitmax(\Pseq)+2\log(4\thitmax(\Pseq))+1+\trelmax(\Pseq)\log \epsilon^{-1}$.
\end{lemma}
To show \cref{lem:uniform_main}, we introduce some terminology.
Let $\pi\in (0,1]^V$ be a positive probability distribution and $f\in \mathbbm{R}^V$ be a vector.
Then, let $\E_\pi(f)\defeq \sum_{v\in V}\pi(v)f(v)=\langle f,\mathbbm{1} \rangle_{\pi}$
and $\Var_\pi(f)\defeq \sum_{v\in V}\pi(v)(f(v)-\E_\pi(f))^2=\langle f,f \rangle_{\pi}-\langle f,\mathbbm{1} \rangle_{\pi}^2$.
Note that, for any probability vector $\mu\in [0,1]^V$, we have
\begin{align*}
d^{(2,\pi)}(\mu)^2
=\sum_{v\in V}\pi(v)\left(\frac{\mu(v)}{\pi(v)}-1\right)^2
=\sum_{v\in V}\pi(v)\left(\frac{\mu(v)}{\pi(v)}\right)^2-1
=\Var_\pi\left(\frac{\mu}{\pi}\right).
\end{align*}
For a transition matrix $P$ such that $\pi(u)P(u,v)=\pi(v)P(v,u)$ holds for all $u,v\in V$,
let 
$\mathcal{E}_{P,\pi}(f)\defeq \langle f,f\rangle_\pi-\langle f,Pf\rangle_\pi= \frac{1}{2}\sum_{u,v\in V}\pi(u)P(u,v)(f(u)-f(v))^2$
be the \emph{Dirichlet form}.
\subsection{Key lemma}
The following key lemma connects the Dirichlet form, $\ell^2$-distance, and hitting time.
\begin{lemma}
\label{lem:dirichlethit}
Let $\mathcal{P}=(P_t)_{t\geq 1}$ be a sequence of irreducible, reversible, and lazy transition matrices.
Suppose that all $P_t$ have the same stationary distribution $\pi$.
Then, for any probability vector $\mu \in [0,1]^V$, 
\begin{align*}
\mathcal{E}_{P,\pi}\left(\frac{\mu}{\pi}\right)
\geq \frac{\Var_\pi\left(\frac{\mu}{\pi}\right)^2}{\thit(P)}.
\end{align*}
\end{lemma}
\begin{proof}

Write $f=\frac{\mu}{\pi}$ and
let $f_{\max}\defeq \max_{v\in V}f(v)$ and $g\defeq f_{\max}\mathbbm{1}-f$.
Since $f(u)-f(v)=g(v)-g(u)$ for any $u,v\in V$, we have
\begin{align*}
\mathcal{E}_{P,\pi}(f)=\frac{1}{2}\sum_{u,v\in V}\pi(u)P(u,v)(f(u)-f(v))^2=\mathcal{E}_{P,\pi}(g)=\langle g,g \rangle_{\pi}-\langle Pg,g\rangle_\pi.
\end{align*}
Let $v_{\max}$ denote a vertex satisfying $f(v_{\max})=f_{\max}$.
Recall that $D_w\in \{0,1\}^{V\times V}$ is a diagonal matrix where $D_w(v,v)\defeq \mathbbm{1}_{v\neq w}$ for all $v\in V$. 
From \cref{lem:MatrixEigen,lem:hiteigen}, we have
\begin{align*}
\langle Pg,g\rangle_\pi
&=\langle D_{v_{\max}}PD_{v_{\max}}g,g\rangle_\pi
\leq \rho(D_{v_{\max}}PD_{v_{\max}})\langle g,g\rangle_\pi
\leq \left(1-\frac{1}{\thit(P)}\right)\langle g,g\rangle_\pi.
\end{align*}
Note that we have $g(u)P(u,v)g(v)=g(u)(D_{v_{\max}}PD_{v_{\max}})(u,v)g(v)$ holds for any $u,v\in V$.
Furthermore,
\begin{align*}
\langle g,g\rangle_\pi
&=\sum_{v\in V}\pi(v)(f_{\max}-f(v))^2
=\sum_{v\in V}\pi(v)(f_{\max}-1+1-f(v))^2\\
&=(f_{\max}-1)^2+\sum_{v\in V}\pi(v)(1-f(v))^2
\geq  \Var_\pi\left(f\right)^2+\Var_\pi\left(f\right).
\end{align*}
The last inequality follows from 
$
0\leq \Var_\pi\left(f\right)
=\sum_{v\in V}\pi(v)\left(\frac{\mu(v)}{\pi(v)}\right)^2-1
\leq 
f_{\max}-1
$.
Therefore, we obtain
\begin{align*}
\mathcal{E}_{P,\pi}(f)
&=\langle g,g \rangle_{\pi}-\langle Pg,g\rangle_\pi
\geq \frac{\langle g,g \rangle_{\pi}}{\thit(P)}
\geq \frac{\Var_\pi\left(f\right)^2}{\thit(P)}.
\end{align*}
\end{proof}
\subsection{Upper bound of mixing time}
To show \cref{lem:uniform_main}, we introduce the following bound shown in \cite{SZ07}.
\begin{lemma}[\cite{SZ07}]
\label{lem:mix_treldecay}
Let $\mathcal{P}=(P_t)_{t\geq 1}$ be a sequence of irreducible, aperiodic, and reversible transition matrices.
Suppose that all $P_t$ have the same stationary distribution $\pi$.
Then, for any probability vector $\mu \in [0,1]^V$, 
\begin{align*}
d^{(2,\pi)}\left(\mu P_{[1,T]}\right)
&\leq d^{(2,\pi)}(\mu)\prod_{t=1}^T\lambda_\star(P_t)
\leq d^{(2,\pi)}(\mu)\exp\left(-\frac{T}{\trelmax(\Pseq)}\right).
\end{align*}
\end{lemma}
\begin{proof}
Applying \cref{lem:normmix} repeatedly to $d^{(2,\pi)}\left(\mu P_{[1,T]}\right)=\left\| \frac{\mu P_{[1,T]}}{\pi} -\mathbbm{1} \right\|_{2,\pi}$, we have 
\begin{align*}
\left\| \frac{\mu P_{[1,T]}}{\pi} -\mathbbm{1} \right\|_{2,\pi}
\leq \left\| \frac{\mu P_{[1,T-1]}}{\pi} -\mathbbm{1}\right\|_{2,\pi}\lambda_\star(P_T) 
\leq \cdots \leq \left\| \frac{\mu}{\pi} -\mathbbm{1}\right\|_{2,\pi} \prod_{t=1}^T\lambda_\star(P_t).
\end{align*}
Then, since
\begin{align*}
\prod_{t=1}^T \lambda_{\star}(P_{t})
= \prod_{i=1}^T \left(1-\frac{1}{\trel(P_t)}\right)
\leq \exp\left(-\sum_{t=1}^T\frac{1}{\trel(P_t)}\right)
\leq \exp\left(-\frac{T}{\trelmax(\Pseq)}\right)
\end{align*}
holds, we obtain the claim.
\end{proof}
Combining \cref{lem:mix_treldecay,lem:dirichlethit}, we obtain the following bounds of $\ell^2$-distance. 
\begin{lemma}
\label{lem:main_mix}
Let $\mathcal{P}=(P_t)_{t\geq 1}$ be a sequence of irreducible, reversible, and lazy transition matrices.
Suppose that all $P_t$ have the same stationary distribution $\pi$.
Then, for any probability distribution $\mu\in [0,1]^V$ and any $0<\epsilon<1$,  
$d^{(2,\pi)}(\mu P_{[1,T]})\leq \epsilon$ holds if
$T\geq  \frac{\mathrm{e}^2}{\mathrm{e}-1}\thitmax(\Pseq)+2\log(4\thitmax(\Pseq))+1+\trelmax(\Pseq)\log \epsilon^{-1}$.
\end{lemma}
\begin{proof}
Since $P$ is reversible with respect to $\pi$, 
\begin{align}
\left(\frac{\mu P}{\pi}\right)(v)
&=\frac{\sum_{u\in V}\mu(u)P(u,v)}{\pi(v)}
=\sum_{u\in V}P(v,u)\frac{\mu(u)}{\pi(u)}
=\left(P\left(\frac{\mu}{\pi}\right)\right)(v)
\label{eq:reversibled2norm}
\end{align}
holds for any $v\in V$, i.e., $\frac{\mu P}{\pi}=P\left(\frac{\mu}{\pi}\right)$ holds.
Combining \cref{eq:reversibled2norm,lem:Mihail,lem:dirichlethit}, we have
\begin{align}
\Var_\pi\left(\frac{\mu P}{\pi}\right)
=\Var_\pi\left(P\left(\frac{\mu}{\pi}\right)\right)
\leq \Var_\pi\left(\frac{\mu}{\pi}\right)-\mathcal{E}_{P,\pi}\left(\frac{\mu}{\pi}\right)
\leq \Var_\pi\left(\frac{\mu}{\pi}\right)\left(1-\frac{\Var_\pi\left(\frac{\mu}{\pi}\right)}{\thit(P)}\right).
\label{eq:decreasevar}
\end{align}
Write $x(i)=\Var_\pi(\mu P_{[1,i]}/\pi)=d^{(2,\pi)}(\mu P_{[1,i]})^2$.
From \cref{eq:decreasevar}, $x$ is non-increasing and $x(t+1)\leq x(t)\left(1-\frac{x(t)}{\thitmax(\Pseq)}\right)$ holds.
Hence, applying \cref{lem:mix_teq_SZ} to $x$, $x(L)\leq 1$ holds if $L\geq \frac{\mathrm{e}^2}{\mathrm{e}-1}\thitmax(\Pseq)+2\log \pi_{\min}^{-1}+1=\left(\frac{\mathrm{e}^2}{\mathrm{e}-1}+o_n(1)\right)\thitmax(\Pseq)$.
Note that $\thit(P)\geq \pi_{\min}^{-1}(1-\pi_{\min})^2\geq \pi_{\min}^{-1}/4$ holds (see e.g., \cite{AF02}).
From \cref{lem:mix_treldecay}, 
\begin{align*}
d^{(2,\pi)}(\mu P_{[1,T]})
&\leq d^{(2,\pi)}(\mu P_{[1,L]})\exp\left(-\frac{T-L}{\trelmax(\Pseq)}\right)
\leq 
\exp\left(-\frac{\trelmax(\Pseq)\log\epsilon^{-1}}{\trelmax(\Pseq)}\right)\leq  \epsilon.
\end{align*}
\end{proof}

\begin{proof}[Proof of \cref{lem:uniform_main}]
Write $P_{[2T,T+1]}=P_{2T}P_{2T-1}\cdots P_{T+1}$ for convenience.
From the reversibility, it is easy to see that $\pi(u)P_{[T+1,2T]}(u,v)=\pi(v)P_{[2T,T+1]}(v,u)$ holds. Hence, we have
\begin{align*}
    \frac{P_{[1,2T]}(u,v)}{\pi(v)}
    &=\frac{\sum_{w\in V}P_{[1,T]}(u,w)P_{[T+1,2T]}(w,v)}{\pi(v)}
    =\sum_{w\in V}\pi(w)\frac{P_{[1,T]}(u,w)}{\pi(w)}\frac{P_{[2T,T+1]}(v,w)}{\pi(w)}\\
    &=\sum_{w\in V}\pi(w)\left(\frac{P_{[1,T]}(u,w)}{\pi(w)}-1\right)\left(\frac{P_{[2T,T+1]}(v,w)}{\pi(w)}-1\right)+1.
\end{align*}
Combining the above and the Cauchy--Schwarz inequality,
\begin{align}
    \left|\frac{P_{[1,2T]}(u,v)}{\pi(v)}-1\right|
    &\leq \sqrt{\sum_{w\in V}\pi(w)\left(\frac{P_{[1,T]}(u,w)}{\pi(w)}-1\right)^2}\sqrt{\sum_{w'\in V}\pi(w')\left(\frac{P_{[2T,T+1]}(v,w')}{\pi(w')}-1\right)^2}\nonumber \\
&=d^{(2,\pi)}\left(P_{[1,T]}(u,\cdotp)\right)d^{(2,\pi)}\left(P_{[2T,T+1]}(v,\cdotp)\right) \label{eq:l2_to_unif}
\end{align}
holds.
Hence, from \cref{lem:main_mix}, we obtain the claim. 
\end{proof}

For completeness, we prove \cref{prop:mixing_SZ07}.
\begin{proof}[Proof of \cref{prop:mixing_SZ07}]
We have 
$d^{(2,\pi)}(\mu)=\sqrt{\sum_{v\in V}\pi(v)\left(\frac{\mu(v)}{\pi(v)}-1\right)^2}\leq \frac{1}{\pi_{\min}}
$ for any probability distribution $\mu$.
Let $T\geq \trelmax(\Pseq)\log(\pi_{\min}^{-1}\epsilon^{-1})$.
Applying \cref{lem:mix_treldecay}, we have $d^{(2,\pi)}(\mu P_{[1,T]})\leq \epsilon$. 
Hence, from \cref{eq:l2_to_unif}, it holds for any $u,v\in V$ that $\left|\frac{P_{[1,2T]}(u,v)}{\pi(v)}-1\right|\leq \epsilon^2$.
\end{proof}

\section{Hitting and cover times} \label{sec:hitting_and_cover_times}
In this section, 
    we consider $k$ independent random walks
    $(X_t(1))_{t\geq 0},\dots,(X_t(k))_{t\geq 0}$
    according to $\Pseq=(P_t)_{t\geq 0}$.
Let $(X_t)_{t\geq 0}$ be a random variable defined as $X_t=(X_t(1),\ldots,X_t(k))\in V^k$.
Let $\tauhit^{(k)}(\Pseq,w)$ (for $w\in V$) and $\taucov^{(k)}(\Pseq)$ be the
random variables denoting hitting and cover times of the $k$ random walks, defined by \cref{eq:hitting_time_def,eq:cover_time_def}.
We bound the expected hitting and cover times: $\thit^{(k)}(\Pseq)$ and $\thit^{(k)}(\Pseq)$ (see \cref{sec:mixing_time_definition} for the definitions).
We sometimes abbreviate $\Pseq$ and write $\tauhit^{(k)}(w)$ and $\taucov^{(k)}$
    if $\Pseq$ is clear from the context.
This section is devoted to the proof of the following results.
\begin{theorem}[Hitting time bound of \cref{thm:main_hitting_cover}] \label{thm:hitting_time}
Let $\mathcal{P}=(P_t)_{t\geq 1}$ be a sequence of irreducible and reversible transition matrices.
Suppose that all $P_t$ have the same stationary distribution $\pi$.
Then, for any $k$,
\begin{align*}
\thit^{(k)}(\Pseq)\leq 20\tsep^{(\pi)}(\Pseq)+\frac{400\thitmax(\Pseq)}{k}.
\end{align*}
\end{theorem}

\begin{theorem}[Cover time bound of \cref{thm:main_hitting_cover}] \label{thm:cover_time}
Let $\mathcal{P}=(P_t)_{t\geq 1}$ be a sequence of irreducible and reversible transition matrices.
Suppose that all $P_t$ have the same stationary distribution $\pi$.
Then, for any $k$,
\begin{align*}
\tcov^{(k)}(\Pseq)\leq 20\tsep^{(\pi)}(\Pseq)+\frac{400\thitmax(\Pseq)\log n}{k}.
\end{align*}
\end{theorem}
The constant factors $20$ and $400$
in \cref{thm:hitting_time,thm:cover_time} may be improved by tuning parameters
but we do not focus on it.

\subsection{Key lemma} \label{sec:key_lemma}
We prove the key result~\cref{lem:HTL}.
\begin{lemma}[Reminder of \cref{lem:HTL}]
\label{lem:HTL_formal}
Let $\mathcal{P}=(P_t)_{t\geq 1}$ be a sequence of irreducible and reversible transition matrices.
Suppose that all $P_t$ have the same stationary distribution $\pi$.
Let $(X_t)_{t\geq 0}$ be a random walk according to $\Pseq=(P_t)_{t\geq 1}$.
Suppose that $X_0$ is sampled from $\pi$.
Then, for any $w\in V$,
\begin{align*}
    \Pr\left[\bigwedge_{t=0}^T\{X_t\neq w\}\right] 
    \leq 
    \prod_{t=1}^T\left(1-\frac{1}{\thit(P_t)}\right)
    \leq
    \exp\left(-\frac{T}{\thitmax(\Pseq)}\right).
\end{align*}
\end{lemma}
\begin{proof}
Recall $D_w\in \{0,1\}^{V\times V}$ is a diagonal matrix defined by $D_w(x,x)\defeq \mathbbm{1}_{x\neq w}$.
From the definition of $D_w$, it is easy to see that
\begin{align*}
    \Pr\left[\bigwedge_{t=0}^T\{X_t\neq w\}, X_T=y\middle|X_0=x\right]=(D_wP_1D_wP_2D_w\cdots D_wP_TD_w)(x,y)
\end{align*}
holds for any $x,y\in V$.
Hence, from the assumption on $X_0$ and the Cauchy--Schwarz inequality, we have
\begin{align}
    \Pr\left[\bigwedge_{t=0}^T\{X_t\neq w\}\right]
    &=\sum_{x\in V}\pi(x) \sum_{y\in V}\left(\prod_{t=1}^T(D_wP_tD_w)\right)(x,y)
    \leq \left\|\left(\prod_{t=1}^T(D_wP_tD_w)\right)\mathbbm{1}\right\|_{2,\pi}.
    \label{eq:DPD1}
\end{align}
From the reversibility of $P_t$, we have $\pi(u)(D_wP_tD_w)(u,v)=\pi(v)(D_wP_tD_w)(v,u)$ for any $u,v\in V$ and $t\geq 1$.
Hence, we can apply \cref{lem:MatrixEigen} repeatedly to \cref{eq:DPD1} and obtain
\begin{align*}
\left\|\left(\prod_{t=1}^T(D_wP_tD_w)\right)\mathbbm{1}\right\|_{2,\pi}
\leq \rho(D_wP_tD_w)\left\|\left(\prod_{t=2}^T(D_wP_tD_w)\right)\mathbbm{1}\right\|_{2,\pi}
\leq \cdots \leq \prod_{t=1}^T\rho(D_wP_tD_w).
\end{align*}
Then, using \cref{lem:hiteigen}, $\rho(D_wP_tD_w)\leq 1-\frac{1}{\thit(P_t)}$ holds for all $t$.
Thus, we obtain the claim.
Note that all $P_t$ are irreducible by assumption.
\end{proof}

\subsection{Upper bound of hitting time}

\begin{lemma} \label{lem:hitting_probability_bound}
Let $\mathcal{P}=(P_t)_{t\geq 1}$ be a sequence of irreducible and reversible transition matrices.
Suppose that all $P_t$ have the same stationary distribution $\pi$.
Then, for any $k$, $x\in V^k$ and $t\geq 0$,
it holds that
\begin{align*}
    \Pr\left[\tauhit^{(k)}(w)\geq \tsep^{(\pi)}(\Pseq)+\frac{20\thitmax(\Pseq)}{k}+t \middle| X_t=x\right] < 0.9.
\end{align*}
\end{lemma}
\begin{proof}
Let $T\defeq\tsep^{(\pi)}(\Pseq)=\tsep^{(\pi)}(\Pseq,1/2)$,
    where $\tsep$ is the separation time defined by \cref{def:tsep}.
Then, the matrix $P_{[1,T]}=\prod_{t=1}^{T}P_t$ can be written as
\begin{align*}
    P_{[1,T]}(u,v) = \frac{1}{2}\pi(u)+\frac{1}{2}Q(u,v)
\end{align*}
for some transition matrix $Q\in[0,1]^{V\times V}$.
Thus, the position of a walker after a transition according to $P_{[1,T]}$
has the same distribution as
the following transition:
 The walker flips a fair coin.
If it is head, the walker moves according to the stationary distribution $\pi$.
Otherwise, the next position of the walker at vertex $u\in V$ is determined by the distribution $Q(u,\cdot)$.

Suppose $k$ independent walkers
    flip their own coins
    and then move to the position
    $X_{T}\in V^k$ according to the transition probability $P_{[1,T]}$.
Let $S\subseteq[k]$ be the random subset of indices of walkers
with a head coin.
Then, the distribution of $X_{T}(i)$ conditioned on $i\in S$ is $\pi$.
Let $w\in V$ be a target vertex.
From \cref{lem:HTL_formal} and the independency of the walkers, for any $T'\geq 0$ and $U\subseteq [k]$, we obtain
\begin{align}
    \Pr\left[\bigwedge_{i\in U}\bigwedge_{T\leq t \leq T+T'}\{X_t(i)\neq w\}\middle| S=U\right] \leq \prod_{i\in U}\prod_{T\leq t \leq T+T'} \left(1-\frac{1}{\thit(P_t)}\right) \leq \exp\left(-\frac{|U|T'}{\thitmax(\Pseq)}\right). \label{eq:hitting_probability}
\end{align}

From the Chernoff inequality (\cref{lem:Chernoff}),
we have $\Pr[|S|<k/4]\leq\exp(-k/16)$ (note that $\E[|S|]=k/2$).
For any events $\mathcal{A}$ and $\mathcal{B}$,
$\Pr[\mathcal{A}]=\Pr[\mathcal{A}\land \mathcal{B}]+\Pr[\mathcal{A}\land \overline{\mathcal{B}}] \leq \Pr[\mathcal{A}|\mathcal{B}]+\Pr[\overline{\mathcal{B}}]$ holds.
Therefore, setting $T'=20\thitmax/k$, we obtain
\begin{align*}
    \Pr[\tauhit^{(k)}(\Pseq)\geq T+T'] &\leq
 \Pr\left[\bigwedge_{i\in [k]}\bigwedge_{T\leq t \leq T+T'}\{X_t(i)\neq y\}\right] \\
 &\leq 
    \Pr\left[\bigwedge_{i\in S}\bigwedge_{T\leq t \leq T+T'}\{X_t(i)\neq y\}\middle| |S|\geq k/4\right] +\Pr[|S|<k/4] \\
    &\leq \exp\left(-\frac{kT'}{4\thitmax}\right)+\exp\left(-\frac{k}{16}\right)\\
    &=\exp(-5)+\exp\left(-\frac{1}{16}\right) < 0.95
\end{align*}
for any $k\geq 1$.
Since this inequality holds
    regardless of the initial position $X_0\in V^k$,
we obtain the claim.
\end{proof}

\begin{proof}[Proof of \cref{thm:hitting_time}]
\Cref{thm:hitting_time} follows from \cref{cor:key_cor,lem:hitting_probability_bound}.
\end{proof}

\subsection{Upper bound of cover time}

\begin{lemma} \label{lem:cover_probability_bound}
Let $\mathcal{P}=(P_t)_{t\geq 1}$ be a sequence of irreducible and reversible transition matrices.
Suppose that all $P_t$ have the same stationary distribution $\pi$.
Then, for any $k$, $x\in V^k$, $t\geq 0$ and every sufficiently large $n$,
it holds that
\begin{align*}
    \Pr\left[\taucov^{(k)}(w)\geq \tsep^{(\pi)}(\Pseq)+\frac{20\thitmax(\Pseq)\log n}{k}+t \middle| X_t=x\right] < 0.95.
\end{align*}
\end{lemma}
\begin{proof}
From \cref{eq:hitting_probability}
    and the union bound over the target vertex $w\in V$,
    we have
\begin{align*}
       \Pr\left[\bigvee_{w\in V}\bigwedge_{i\in U}\bigwedge_{T\leq t \leq T+T'}\{X_t(i)\neq w\}\middle| S=U\right] 
       &\leq n \prod_{i\in U}\prod_{T\leq t \leq T+T'} \left(1-\frac{1}{\thit(P_t)}\right) \\
       &\leq n\exp\left(-\frac{|U|T'}{\thitmax(\Pseq)}\right).
\end{align*}
Setting $T'=20\thitmax\log n/k$, we obtain
\begin{align*}
    \Pr[\taucov^{(k)}(\Pseq)> T+T'] &\leq
 \Pr\left[\bigvee_{w\in V}\bigwedge_{i\in [k]}\bigwedge_{T\leq t \leq T+T'}\{X_t(i)\neq y\}\right] \\
 &\leq 
    \Pr\left[\bigvee_{w\in V}\bigwedge_{i\in S}\bigwedge_{T\leq t \leq T+T'}\{X_t(i)\neq y\}\middle| |S|\geq k/4\right] +\Pr[|S|<k/4] \\
    &\leq n\exp\left(-\frac{kT'}{4\thitmax}\right)+\exp\left(-\frac{k}{16}\right)\\
    &=n\exp(-5\log n)+\exp\left(-\frac{k}{16}\right) \\
    &< 0.94+O(n^{-4})<0.95
\end{align*}
for any $k\geq 1$ and every sufficiently large $n$.
Since this inequality holds
    regardless of the initial position $X_0\in V^k$,
we obtain the claim.
\end{proof}

\begin{proof}[Proof of \cref{thm:cover_time}]
\Cref{thm:cover_time} follows from \cref{cor:key_cor,lem:cover_probability_bound}.
\end{proof}

\section{Meeting and coalescing times}
\label{sec:coalescing_time}
We show \cref{prop:exponential_lower_meet,thm:main_coalescing} in this section.
\subsection{Key lemma}
\begin{lemma}[Meeting Time Lemma on dynamic graphs; Reminder of \cref{lem:MTL}] \label{lem:MTL_formal}
Let $\mathcal{P}=(P_t)_{t\geq 1}$ be a sequence of irreducible, reversible, and lazy transition matrices.
Suppose that all $P_t$ have the same stationary distribution $\pi$.
Let $(X_t)_{t\geq 0}$ be a random walk according to $\Pseq=(P_t)_{t\geq 1}$.
Suppose that $X_0$ is sampled from $\pi$.
Then, for any sequence $w_0,w_1,\dots,w_T$ of vertices,
\begin{align*}
    \Pr\left[\bigwedge_{t=0}^T\{X_t\neq w_t\}\right] \leq 
    \prod_{t=1}^T\left(1-\frac{1}{\thit(P_t)}\right)
    \leq \exp\left(-\frac{T}{\thitmax(\Pseq)}\right).
\end{align*}
\end{lemma}
\begin{proof}
Recall that $D_w\in \{0,1\}^{V\times V}$ is a diagonal matrix defined by $D_w(x,x)=\mathbbm{1}_{x\neq w}$.
In the same
    way as the proof of \cref{lem:HTL_formal}, 
\begin{align*}
    \Pr\left[\bigwedge_{t=0}^T\{X_t\neq w_t\},X_T=y\middle|X_0=x\right]
    &=(D_{w_0}P_1D_{w_1}P_2D_{w_2}\cdots D_{w_{T-1}}P_TD_{w_{T}})(x,y)
\end{align*}
holds for any $x,y\in V$, and hence we have 
\begin{align*}
    \Pr\left[\bigwedge_{t=0}^T\{X_t\neq w_t\}\right]
    &=\sum_{x\in V}\pi(x)\sum_{y\in V}\left(\prod_{t=1}^T(D_{w_{t-1}}P_t D_{w_t})\right)(x,y)
    \leq \left\|\left(\prod_{t=1}^T(D_{w_{t-1}}P_tD_{w_t})\right)\mathbbm{1}\right\|_{2,\pi}.
\end{align*}
Here, we used the Caushy--Schwarz inequality.
Since $P_t$ is reversible and lazy,
applying \cref{lem:normmeet} repeatedly yields
\begin{align*}
\left\|\left(\prod_{t=1}^T(D_{w_{t-1}}P_tD_{w_t})\right)\mathbbm{1}\right\|_{2,\pi}
&\leq \sqrt{\rho(D_{w_0}P_1D_{w_0})\rho(D_{w_1}P_1D_{w_1})}\left\|\left(\prod_{t=2}^T(D_{w_{t-1}}P_tD_{w_t})\right)\mathbbm{1}\right\|_{2,\pi}\\
&\leq \cdots \leq \prod_{t=1}^T\sqrt{\rho(D_{w_{t-1}}P_tD_{w_{t-1}})\rho(D_{w_t}P_tD_{w_t})}.
\end{align*}
Finally, from \cref{lem:hiteigen}, we have $\rho(D_wP_tD_w)\leq 1-\frac{1}{\thit(P_t)}$ for any $t$.
Thus, we obtain the claim.
Note that $P_t$ is irreducible for any $t$.
\end{proof}

\subsection{Upper bound of coalescing time}
Consider the coalescing random walks
$(\Co_t(1))_{t\geq 0}, (\Co_t(2))_{t\geq 0}, \ldots, (\Co_t(n))_{t\geq 0}$  according to $\Pseq=(P_t)_{t\geq 1}$
(see \cref{sec:mixing_time_definition} for the definition).
This section is devoted to the proof of
\cref{thm:main_coalescing}.
\begin{theorem}[Precise statement of \cref{thm:main_coalescing}] \label{thm:coalescing_time_bound}
Let $\mathcal{P}=(P_t)_{t\geq 1}$ be a sequence of irreducible, reversible, and lazy transition matrices.
Suppose that all $P_t$ have the same stationary distribution $\pi$.
Then, $\tcoal(\Pseq)\leq C\thitmax(\Pseq)$ holds for some positive constant $C$.
\end{theorem}
To prove \cref{thm:coalescing_time_bound}, we introduce the following notation. 
Let $\ell_x\defeq \lceil \log_2(x)\rceil$ and $\ell\defeq \ell_n$.
Let $K$ be a suitable constant that will be determined later.
Define $L_\ell, L_{\ell-1}, \ldots, L_0$ recursively 
by $L_\ell=\tsep^{(\pi)}(\Pseq)$ and $L_{i}=L_{i+1}+\left\lceil \frac{K\thitmax(\Pseq)}{2^i}\right\rceil$.
In other words,
$L_i=\tsep^{(\pi)}(\Pseq)+\sum_{j=i}^{\ell-1}\left\lceil \frac{K\thitmax(\Pseq)}{2^j}\right\rceil$ and 
$L_0=\tsep^{(\pi)}(\Pseq)+\sum_{j=0}^{\ell-1}\left\lceil \frac{K\thitmax(\Pseq)}{2^j}\right\rceil\leq \tsep^{(\pi)}(\Pseq)+\log_2(n)+2K\thitmax(\Pseq)$.

The following result states a relation between the coalescing random walk and the independent random walks.
The proof is
essentially based on the argument of \cite{Oliveira12}.
\begin{lemma}
\label{lem:coal_mult}
Let $(X_t)_{t\geq 0}=((X_t(1),\ldots,X_t(n)))_{t\geq 0}$ be $n$ independent random walks, where each walker is according to $\Pseq=(P_t)_{t\geq 1}$.
Suppose that $\Co_0=x$ and $X_0=x$ for an arbitrary initial position of walkers $x=(x_1, x_2, \ldots, x_n)\in V^n$.
Then, we have
\begin{align*}
    \Pr\left[\taucoal(\Pseq)>L_0\right]
    \leq 
    \Pr\left[\bigvee_{a=2}^n\bigwedge_{b\in [2^{\ell_a}-1]}\bigwedge_{t\in [L_{\ell_b+1},L_{\ell_b}]}\{X_t(b)\neq X_t(a)\}\right].
\end{align*}
\end{lemma}
\begin{proof}
First, we define the \emph{random walks with killings} $(Y_t)_{t\geq 0}=((Y_t(1),\ldots,Y_t(n)))_{t\geq 0}$, where a walker is killed when it meets another walker with a smaller index.
Formally, $(Y_t)_{t\geq 0}$ is a
    Markov chain on a state space $(V\cup \{\partial\})^n$,
    where $\partial\not\in V$ is a \emph{coffin state}.
Set $Y_t(1)=X_t(1)$ for all $t\geq 0$.
For $t\geq 0$ and $a\geq 1$, define $Y_t(a)$ inductively as follows.
Suppose that $(Y_s)_{s=0}^{t-1}$ and $Y_t(1), \ldots, Y_t(a-1)$ are determined.
Then, define
\begin{align*}
\tau(a)&\defeq \min\{s\geq 0: \textrm{$X_s(a)=Y_s(b)$ for some $b<a$}\}, \hspace{1em} \mathrm{and}\\
Y_t(a)&\defeq 
\begin{cases}
X_t(a) &\text{if $t<\tau(a)$},\\
\partial &\text{otherwise}.
\end{cases}
\end{align*}
%

Next, we define the \emph{random walks with a list of allowed killings} $(Z_t)_{t\geq 0}=((Z_t(1),\ldots,Z_t(n)))_{t\geq 0}$, where  
a walker $a\geq 2^{i+1}$ is killed when it meets a ``killer'' walker $2^i\leq b<2^{i+1}$ during the time period $L_{i+1}< t\leq L_i$.
\Cref{fig:coalescing} is an overall picture.
\begin{figure}
    \centering
    \includegraphics[width=\hsize]{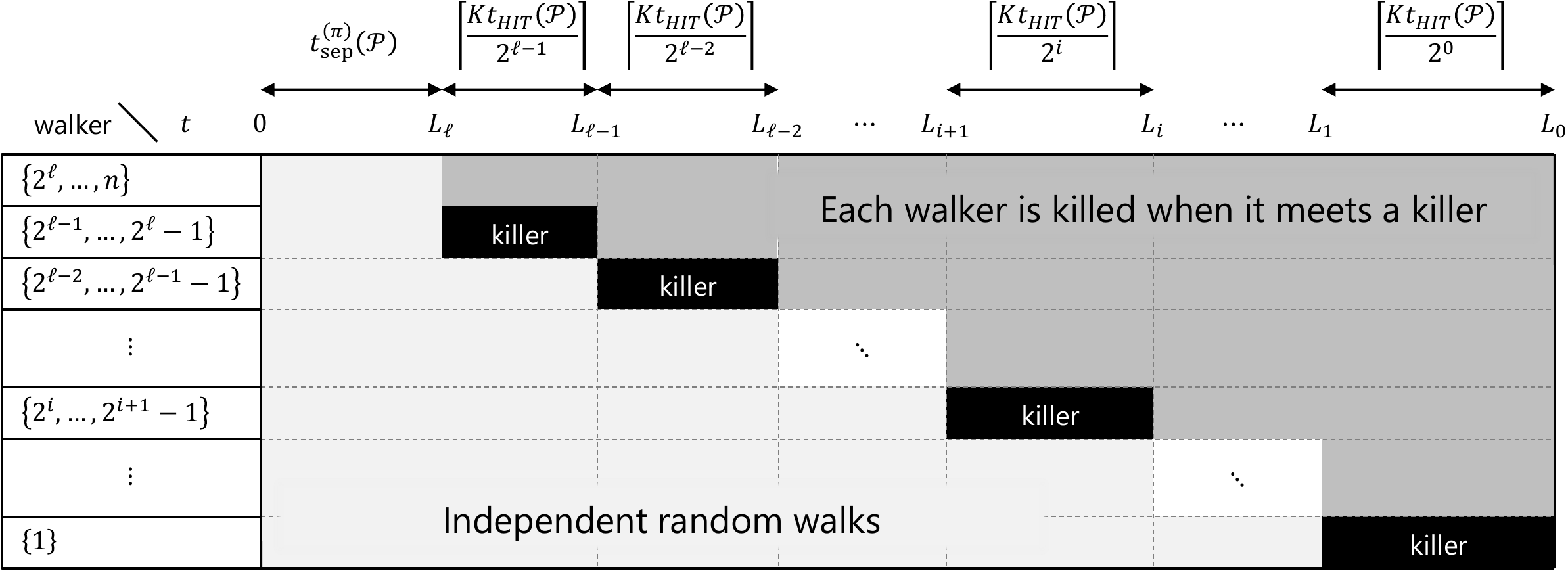}
    \caption{Random walk with a list of allowed killings (\cite{Oliveira12}).}
    \label{fig:coalescing}
\end{figure}
Formally, define the list of allowed killings $\mathcal{A}=(A_t)_{t\geq 0}$ as 
$A_t\defeq \emptyset$ for $t\leq L_{\ell}$ and  
$A_t\defeq \left \{(b,a): 2^i\leq b<2^{i+1}, a\geq 2^{i+1} \right\}$ for $L_{i+1}< t\leq L_i$ with $0\leq i<\ell$.
Let $Z_t(1)=X_t(1)$ for all $t\geq 0$.
For $t\geq 0$ and $a\geq 1$, $Z_t(a)$ is inductively defined as follows.
Suppose that $(Z_s)_{s=0}^{t-1}$ and $Z_t(1), \ldots, Z_t(a-1)$ are determined.
Then, let
\begin{align*}
\tau^\mathcal{A}(a)&\defeq \min\{s\geq 0: \textrm{$X_s(a)=Z_s(b)$ for some $(b,a)\in A_s$}\}, \hspace{1em} \mathrm{and}\\
Z_t(a)&\defeq 
\begin{cases}
X_t(a) &\text{if $t< \tau^\mathcal{A}(a)$,}\\
\partial &\text{otherwise}.
\end{cases}
\end{align*}

For a vector $x\in (V\cup \{\partial\})^n$, let $\mathcal{S}(x)=\{i\in [n]: x(i)\neq \partial\}$.
Obviously, $\Pr[|S(\Co_t)|\geq z|\Co_0= x]=\Pr[|\mathcal{S}(Y_t)|\geq z|X_0=x]$ holds for any $t\geq 0$, $z\geq 0$ and $x\in V^n$.
Furthermore, $\Pr[|\mathcal{S}(Y_t)|\geq z|X_0=x]\leq \Pr[|\mathcal{S}(Z_t)|\geq z|X_0=x]$ holds for any $t\geq 0$, $z\geq 0$ and $x\in V^n$.
To see this, consider using $(\Co_t)_{t\geq 0}$ instead of $(X_t)_{t\geq 0}$ in definitions of both $(Y_t)_{t\geq 0}$ and $(Z_t)_{t\geq 0}$.
Then, $|\mathcal{S}(Y_t)|\leq |\mathcal{S}(Z_t)|$ holds for any $t\geq 0$.
Note that $\Pr[\Co_t(i)=v|\Co_0(i)=u]=\Pr[X_t(i)=v|X_0(i)=u]$ holds for any $u,v\in V$ and $i\in [n]$.
From the definition of the random walk with allowed killings, $Z_t(i)=X_t(i)$ until it meets a killer.
Hence, we have
\begin{align*}
&\Pr[\taucoal(\Pseq)>L_0]\\
&\leq \Pr\left[|\mathcal{S}(Z_{L_0})|\geq 2\right] 
=\Pr\left[\bigvee_{a=2}^{n}\{\textrm{Walker $a$ is not killed}\}\right]\\
&=\Pr\left[\bigvee_{a=2}^{n}\bigwedge_{b=1}^{2^{\ell_a}-1}\bigwedge_{t=L_{\ell_b+1}}^{L_{\ell_b}}\{Z_t(a)\neq Z_t(b)\}\right]
=\Pr\left[\bigvee_{a=2}^{n}\bigwedge_{b=1}^{2^{\ell_a}-1}\bigwedge_{t=L_{\ell_b+1}}^{L_{\ell_b}}\{X_t(b)\neq X_t(a)\}\right].
\end{align*}
\end{proof}
Combining \cref{lem:coal_mult,lem:MTL_formal}, we obtain the following lemma.
\begin{lemma}
\label{lem:coalcons}
Let $\mathcal{P}=(P_t)_{t\geq 1}$ be a sequence of irreducible, reversible, and lazy transition matrices.
Suppose that all $P_t$ have the same stationary distribution $\pi$.
Then, for any $x\in V^{n}$, 
$
\Pr\left[\taucoal(\Pseq)>T\middle|\Co_0=x\right]\leq 1-10^{-5}
$ holds if $T\geq \tsep^{(\pi)}(\Pseq)+80\thit(\Pseq)+\log_2(n)$.
\end{lemma}
\begin{proof}
%
From the definition \cref{def:tsep} of the separation time $\tsep^{(\pi)}(\Pseq)=L_\ell$, 
there is a transition matrix $Q\in [0,1]^{V\times V}$ such that
\begin{align*}
P_{[1,L_\ell]}(x,u)
=
\frac{1}{2}\pi(u)+\frac{1}{2}Q(x,u)
\end{align*}
holds for all $x,u\in V$. 
Hence, the distribution of $n$ walkers $X_{L_\ell}(1),\ldots,X_{L_\ell}(n)$ can be simulated as follows:
Each walker $i\in [n]$ flips its own fair coin.
If it is head, the walker's position $X_{L_\ell}(i)$ is sampled according to $\pi$. Otherwise, it is sampled according to the distribution $Q(X_0(i),\cdotp)$.
Let $I\subseteq [n]$ denote a random subset of indices with a head coin.
Let $\mathcal{W}\defeq \{W\subseteq [n]: \textrm{$|\{2^i,\ldots,2^{i+1}-1\}\cap W|\geq 2^i/4$ holds for all $0\leq i<\ell$}\}$ be a set of subsets of $[n]$.
Then, from \cref{lem:coal_mult}, 
\begin{align}
&\Pr\left[\taucoal(\Pseq)>L_0\right] \nonumber \\
&\leq \sum_{W\subseteq [n]}\Pr\left[\bigvee_{a=2}^n\bigwedge_{b\in [2^{\ell_a}-1]}\bigwedge_{t\in [L_{\ell_b+1},L_{\ell_b}]}\{X_t(b)\neq X_t(a)\}\middle|I=W\right]\Pr[I=W]\nonumber \\
    &\leq \Pr[I\notin \mathcal{W}]
    +\max_{W\in \mathcal{W}}\Pr\left[\bigvee_{a=2}^n\bigwedge_{b\in [2^{\ell_a}-1]\cap W}\bigwedge_{t\in [L_{\ell_b+1},L_{\ell_b}]}\{X_t(b)\neq X_t(a)\}\middle|\bigwedge_{i\in W}\{X_{L_\ell}(i)\sim \pi\} \right].
    \label{eq:coal_bound}
\end{align}
%
Let $I_j=\mathbbm{1}_{j\in I}\in\{0,1\}$ denote the
    binary indicator for the random subset $I$.
For the first term of \cref{eq:coal_bound}, applying the Chernoff inequality (\cref{lem:Chernoff}) yields
\begin{align}
\Pr[I\in \mathcal{W}]
&=\Pr\left[\bigwedge_{i=0}^{\ell-1}\{|I\cap \{2^i,\ldots,2^{i+1}-1\}|\geq 2^i/4\}\right] =\prod_{i=0}^{\ell-1}\Pr\left[\sum_{j=2^i}^{2^{i+1}-1}I_j\geq 2^i/4\right] \nonumber\\
&\geq
    \prod_{i=0}^{5}\left(1-\exp\left(-\frac{2^i}{16}\right)\right)\cdot    
    \prod_{i=6}^{\ell-1}\left(1-\exp\left(-\frac{2^i}{16}\right)\right) \nonumber\\
&\geq 0.00033\cdot  \prod_{i=2}^{\ell-5}\left(1-\exp\left(-2^{i}\right)\right) \geq 0.00033\cdot \prod_{i=2}^{\infty}\left(1-2^{-i}\right)\nonumber\\
&\geq 0.00033\cdot \left(1-\sum_{i=2}^\infty 2^{-i}\right) \geq 0.00016.\label{eq:coal_expbound}
\end{align}
Note that $\E\left[\sum_{j=2^i}^{2^{i+1}-1}I_j\right]=2^i/2$.
Next, we bound the second term of \cref{eq:coal_bound}.
For any $W\in \mathcal{W}$, 
\begin{align*}
&\Pr\left[\bigvee_{a=2}^n\bigwedge_{b\in [2^{\ell_a}-1]\cap W}\bigwedge_{t\in [L_{\ell_b+1},L_{\ell_b}]}\{X_t(b)\neq X_t(a)\}\middle|\bigwedge_{i\in W}\{X_{L_\ell}(i)\sim \pi\}\right]\\
&\leq \sum_{a=2}^n\prod_{b\in [2^{\ell_a}-1]\cap W}\Pr\left[\bigwedge_{t\in [L_{\ell_b+1},L_{\ell_b}]}\{X_t(b)\neq X_t(a)\}\middle|\bigwedge_{i\in W}\{X_{L_\ell}(i)\sim \pi\}\right]
\end{align*}
holds.
Applying the meeting time lemma (\cref{lem:MTL_formal}),
for any $b\in W$, we have
\begin{align*}
\Pr\left[\bigwedge_{t\in [L_{\ell_b+1},L_{\ell_b}]}\{X_t(b)\neq X_t(a)\}\middle|\bigwedge_{i\in W}\{X_{L_\ell}(i)\sim \pi\}\right]
\leq \exp\left(-\frac{L_{\ell_b}-L_{\ell_b+1}}{\thitmax(\Pseq)}\right)
\leq \exp\left(-\frac{K}{2^{\ell_b}}\right).
\end{align*}
Hence, for any $W\in \mathcal{W}$, 
\begin{align}
&\Pr\left[\bigvee_{a=2}^n\bigwedge_{b\in [2^{\ell_a}-1]\cap W}\bigwedge_{t\in [L_{\ell_b+1},L_{\ell_b}]}\{X_t(b)\neq X_t(a)\}\middle|\bigwedge_{i\in W}\{X_{L_\ell}(i)\sim \pi\} \right] \nonumber \\
&\leq \sum_{i=1}^{\ell-1}\sum_{a=2^i}^{2^{i+1}-1}\prod_{j=1}^{i-1}\prod_{b\in \{2^j,\ldots,2^{j+1}-1\}\cap W}\exp\left(-\frac{K}{2^{j}}\right) 
\leq \sum_{i=1}^{\ell-1}\sum_{a=2^i}^{2^{i+1}-1}\prod_{j=1}^{i-1}\exp\left(-\frac{K}{4}\right) \nonumber \\
&\leq \sum_{i=1}^{\ell-1}2^i\exp\left(-\frac{K}{4}i\right)\leq \frac{2}{\mathrm{e}^{K/4}-2}.
\label{eq:coal_mainbound}
\end{align}
Combining \cref{eq:coal_bound,eq:coal_expbound,eq:coal_mainbound} 
with $K=40$, 
$
\Pr\left[\taucoal(\Pseq)>T\middle|\Co_0=x\right]
\leq 1-0.00016+\frac{2}{\mathrm{e}^{10}-2}\leq 1-10^{-5}
$ holds.
\end{proof}
\begin{proof}[Proof of \cref{thm:coalescing_time_bound}]
Let $T\defeq \lceil (85+o_n(1))\thitmax(\Pseq)\rceil$.
\Cref{lem:coalcons} implies that, for any $t\geq 0$ and $x\in V^n$, $\Pr[\taucoal(\Pseq)>t+T|\Co_t=x]\leq 1-10^{-5}$ holds.
Thus we obtain the claim from \cref{cor:key_cor}.
\end{proof}

\subsection{Lower bound of meeting time}
\label{sec:meeting_time_lower_bound}
In this section,
we prove \cref{prop:exponential_lower_meet}.
\begin{proof}[Proof of \cref{prop:exponential_lower_meet}]
We consider the lazy simple random walk
    on the graph sequence given in Proposition 12 of~\cite{OT11}.
For completeness, we present
    the sequence formally.
For a graph $H$ and a permutation $\eta$ on $V(H)$,
let $\eta(H)$ be the graph
given by
$V(\eta(H))=V(H)$ and
$E(\eta(H))=\{\{\eta(u),\eta(v)\}:\{u,v\}\in E(H)\}$.

For any integer $m\in\mathbb{N}$,
let $U=\{u_0,\dots,u_{m-1}\}$ and $W=\{w_0,\dots,w_{m-1}\}$.
Define the graph $G$ by
$V(G)\defeq U\cup W$ and
\begin{align}
    E(G)\defeq\{u_0,w_0\}\cup \bigcup_{i=1}^{m-1} \{u_i,u_0\} \cup \bigcup_{j=1}^{m-1} \{w_j,w_0\}. \label{eq:graph_def}
\end{align}
Let $\eta$ be the permutation on $V$ defined by
    $\eta(u_i) = u_{(i+1)\bmod m}$
    and
    $\eta(w_i) =
    w_{(i+1)\bmod m}$.
    
\begin{figure}[tbp]
\center
\begin{tikzpicture}[
           every node/.style={
			circle,
               draw,
text centered, 
			}]
\def\D{0.6cm}
\def\S{3cm}
\node (u0) {$u_0$};
\node[above = \D of u0] (u1) {$u_1$};
\node[left = \D of u0] (u2) {$u_2$};
\node[below = \D of u0] (u3) {$u_3$};
\node[right = \D of u0] (v0) {$w_0$};
\node[above = \D of v0] (v1) {$w_1$};
\node[right = \D of v0] (v2) {$w_2$};
\node[below = \D of v0] (v3) {$w_3$};

\draw[-](u0)--(u1);
\draw[-](u0)--(u2);
\draw[-](u0)--(u3);
\draw[-](u0)--(v0);
\draw[-](v0)--(v1);
\draw[-](v0)--(v2);
\draw[-](v0)--(v3);

\node[right = \S of v2] (u1_) {$u_1$};
\node[above = \D of u1_] (u2_) {$u_2$};
\node[left = \D of u1_] (u3_) {$u_3$};
\node[below = \D of u1_] (u0_) {$u_0$};
\node[right = \D of u1_] (v1_) {$w_1$};
\node[above = \D of v1_] (v2_) {$w_2$};
\node[right = \D of v1_] (v3_) {$w_3$};
\node[below = \D of v1_] (v0_) {$w_0$};

\draw[-](u1_)--(u2_);
\draw[-](u1_)--(u3_);
\draw[-](u1_)--(u0_);
\draw[-](u1_)--(v1_);
\draw[-](v1_)--(v2_);
\draw[-](v1_)--(v3_);
\draw[-](v1_)--(v0_);

\end{tikzpicture}
\caption{The graph $G$ given by \cref{eq:graph_def} for $m=4$ (left) and $\eta(G)$ (right). \label{fig:bad_graph}}
\end{figure}
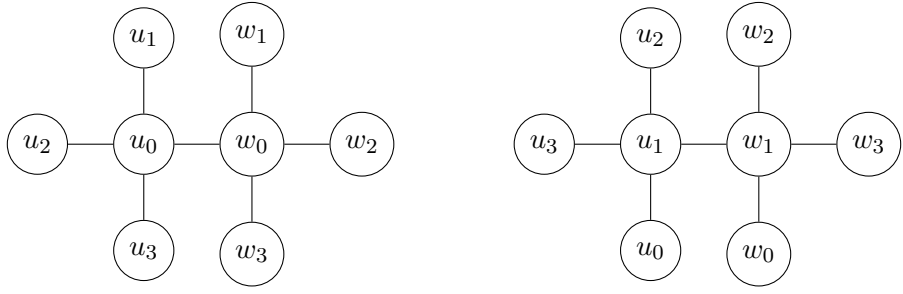

We claim that the lazy simple random walk on the sequence $(G_t)_{t\geq 1}$ given by $G_1=G$ and $G_{t+1}=\eta(G_t)$ ($t\geq 1$) has the desired property (see \cref{fig:bad_graph}).
Consider two independent lazy simple random walks $(X_t(1))_{t\geq 0}$ and $(X_t(2))_{t\geq 0}$ with initial positions $(X_0(1),X_0(2))=(u_{m-1},w_{m-1})$.
Suppose $\taumeet\leq t$.
Then, there exists $t'\leq t$ such that
either $X(1)$ or $X(2)$
moves along the edge $\{u_j,w_j\}$ for $j=t'\bmod m$.
Focus on the walk $X(1)$
and
suppose $X_{t'-1}(1)\in U$ and $X_{t'}(1)\in W$.
To reach $u_j$, the walker $X(1)$ must choose the self loop for $m-1$ consecutive times,
which occurs with probability $2^{-m+1}$.
Therefore, by the union bound over $X(0)$ and $X(1)$, we have
$\Pr[\taumeet \leq t]\leq 2\cdot 2^{-m+1}$
and we have
$\tmeet = 2^{\Omega(m)}$.
\end{proof}
\section{Pull voting} \label{sec:pull_voting}

In this section,
    we prove \cref{thm:consensus_time}.
Our proof of bounding $\E[\taucons]$ is inspired by the idea of well-known \emph{duality} between the pull voting and coalescing random walk~\cite{HP01}.

\begin{proposition}[Duality in static setting~\cite{HP01}] \label{prop:duality_static}
Let $P\in[0,1]^{V\times V}$
    be an irreducible transition matrix.
Let $\taucons(P)$ be the consensus time of the pull voting according to $P$ where all vertices initially hold distinct opinions.
Let $\taucoal(P)$ be
the coalescing time of the coalescing random walk according to $P$.
Then, for every $j\geq 0$, $\Pr[\taucoal(P)\leq j]=\Pr[\taucons(P)\leq j]$ holds.
\end{proposition}

From \cref{prop:duality_static},
we can obtain bounds of $\taucons$
by studying $\taucoal$.
Indeed, the proof of previous results bounding $\taucons$ on a static graph relies on the duality.
In this paper, we obtain the following
    consensus--coalescing relation
    that is analogous
    to
    \cref{prop:duality_static}
    in the time inhomogeneous setting.

\begin{proposition}[Consensus--Coalescing relation on dynamic graphs] \label{prop:duality_dynamic}
Let $\Pseq$ be a sequence of transition matrices (not necessarily has a time-homogeneous stationary distribution).
Consider the pull voting
    according to $\Pseq$
    such that initially all vertices have $n$ distinct opinions.
Then, there is a sequence $(\mathcal{Q}^{(i)})_{i\geq 0}$ where each $\mathcal{Q}^{(i)}=(Q^{(i)}_t)_{t\geq 1}$ is a transition matrix sequence
    such that,
    for every $j\geq 0$,
    \begin{align*}
        \Pr[\taucons(\Pseq)\leq j]=\Pr[\taucoal(\mathcal{Q}^{(j)})\leq j]
    \end{align*}
    holds.
Moreover, if $\Pseq$ is
    reversible
    and has the time-invariant
    stationary distribution $\pi$,
    then so do the sequences $\mathcal{Q}^{(i)}$ for all $i$.
\end{proposition}
Indeed, if $\Pseq=(P)_{t\geq 1}$ is a sequence of
a static transition matrix,
    then $\mathcal{Q}^{(i)}=\Pseq$
        for all $i\geq 0$,
        implying \cref{prop:duality_static}.

The proof of the duality theorem in the static setting (\cref{prop:duality_static})
    is obtained by constructing
    a coupling of the pull voting and the
    coalescing random walk
    with equal consensus and coalescing times.
Our proof of \cref{prop:duality_dynamic}
    is based on essentially the same argument.
In \cref{sec:exponential_lower_bound_pull}, we present a
    sequence of graphs on which
    the pull voting according to $\LSimP$ on it has an exponential consensus time as follows:
\begin{proposition} \label{prop:exponential_lower_bound_pull}
    There is a sequence $(G_t)_{t\geq 1}$ of graphs
    on which the pull voting
    according to $\Pseq=(\LSimP(G_t))_{t\geq 1}$
    over opinion set $\Sigma=\{0,1\}$
    satisfies $\tcons=2^{\Omega(n)}$.
\end{proposition}

\subsection{Consensus--Coalescing relation}
We prove \cref{prop:duality_dynamic}.
The proof is essentially based on
    the notion of linear voter model of \cite{CR16}.

\begin{proof}[Proof of \cref{prop:duality_dynamic}.]
Let $P\in[0,1]^{V\times V}$ be a
    transition matrix and
    $\mathcal{S}$ be the set of all binary $V\times V$ matrices such that each row contains
    exactly one $1$.
For each matrix $S\in\mathcal{S}$,
we define a probability distribution $\mu_P$ over $\mathcal{S}$ by
\begin{align*}
    \mu_P(S)=\prod_{(i,j)\in V\times V: S_{i,j}=1} P_{i,j}
\end{align*}
for each $S\in\mathcal{S}$.
We interpret $S\in \mathcal{S}$ as the
    list of selections at a specific round of the pull voting:
    Specifically, $s_{i,j}=1$ if and only if $i$ selects $j$ at the pull voting.
Then, $\mu_P(\cdot)$ can be seen as the probability distribution
    over the set of all possible selection lists during the pull voting according to $P$.

Given a sequence $S_1,\dots,S_i\in \mathcal{S}$ of $i$ matrices,
    we can simulate the pull voting for $i$ rounds as follows.
Let $y_0\in\Sigma^V$ denote the initial opinion configuration
    where each vertex has a distinct opinion.
Then, $y_i = S_i y_{i-1} = S_i S_{i-1} \dots S_1 y_0 = \prod_{t=1}^i S_{i-t+1}y_0$.
We say that an opinion vector $y\in\Sigma^V$ \emph{is in consensus} if $y=\sigma\mathbbm{1}$ holds for some $\sigma\in \Sigma$.
For fixed $i\in\Nat$, let
\begin{align*}
    \mathcal{S}^{(i)}_{\mathrm{cons}} =
    \{(S_1,\dots,S_i)\in \mathcal{S}^i:\text{$y_i$ is in consensus}\}.
\end{align*}
Here, note that, if $y_{i-1}$ is in consensus,
    then so does $y_i=S_iy_{i-1}$.
Then, we have
\begin{align}
    \Pr[\taucons(\Pseq)\leq i] = \sum_{(S_1,\dots,S_i)\in\mathcal{S}^{(i)}_{\mathrm{cons}}} \prod_{t=1}^i \mu_{P_t}(S_t). \label{eq:taucons}
\end{align}

We show that, for every $i$, there
is a sequence $\mathcal{Q}=\mathcal{Q}^{(i)}=(Q^{(i)}_t)_{t\geq 0}$ of transition matrices
such that
the right hand side of \cref{eq:taucons} is equal to $\Pr[\taucoal(\mathcal{Q}^{(i)})\leq i]$.
Consider the sequence $\mathcal{Q}^{(i)}$
    defined by
    \begin{align*}
        Q^{(i)}_t = \begin{cases}
        P_{i-t+1} & \text{if $t\leq i$},\\
        P_0 & \text{if $t>i$}.
        \end{cases}
    \end{align*}
Note that, if $\Pseq$ is reversible and has a time-invariant stationary distribution $\pi$, then so does $\mathcal{Q}^{(i)}$ for every $i$.

Fix $i$ and consider the coalescing random walk
    according to $\mathcal{Q}^{(i)}$.
We call a vector $c\in\mathbb{Z}_{\geq 0}^V$ satisfying $\sum_{v\in V}c_i=|V|$ a \emph{walker configuration vector}.
A walker configuration vector can be
    interpreted as the vector
    denoting the
    number of walkers on vertices, i.e., $c_v$ is the number of walkers on $v$.
Given $S'_1,\dots,S'_i\in\mathcal{S}$,
    we can
    simulate the coalescing random walk for $i$ rounds as follows:
Let $c_0=\mathbbm{1}^{\top}$ (initially, each vertices has exactly one walker).
For $1\leq j\leq i$, let $c_j=c_{j-1}S'_j = c_0\prod_{k=1}^j S'_k$.
Intuitively speaking, the matrix $S'_j$ denotes the transition result: $(S'_j)_{u,v}=1$ if and only if $u$ sends \emph{all} walkers on it to $v$ at the $j$-th round.
A vertex $v\in V$ at round $i$ has a walker
    if and only if $(c_i)_v>0$.
    
We say that a walker configuration vector $c$ \emph{is in coalesce} if $c=|V|e_v$ for some $v\in V$, where $e_v\in\{0,1\}^V$ is the binary indicator vector for a vertex $v$ (i.e., $(e_v)_u=1$ if and only if $v=u$).
Note that, if $c_{j-1}$ is in coalesce,
then so does $c_j=c_{j-1}S'_j$.
For a coalescing random walk is according to a transition matrix $Q$, a transition result $S'\in\mathcal{S}$ occurs with probability $\mu_Q(S')$.
Let
\begin{align*}
    \mathcal{S}^{(i)}_{\mathrm{coal}}=\{(S'_1,\dots,S'_i)\in\mathcal{S}^i:\text{$c_i$ is coalescing}\}.
\end{align*}
Then, we have
\begin{align}
    \Pr[\taucoal(\mathcal{Q}^{(i)})\leq i] = \sum_{(S'_1,\dots,S'_i)\in\mathcal{S}^{(i)}_{\mathrm{coal}}} \prod_{t=1}^i\mu_{Q^{(i)}_t}(S'_t) 
    = \sum_{(S'_1,\dots,S'_i)\in\mathcal{S}^{(i)}_{\mathrm{coal}}}
    \prod_{t=1}^i \mu_{P_{t}}(S'_{i-t+1}). \label{eq:taucoal}
\end{align}

We compare \cref{eq:taucons,eq:taucoal}.
Indeed, it holds that  $(S_1,\dots,S_i)\in\mathcal{S}^{(i)}_{\mathrm{cons}}$ if and only if $(S_i,\dots,S_1)\in\mathcal{S}^{(i)}_{\mathrm{coal}}$.
To see this,
suppose $(S_1,\dots,S_i)\in\mathcal{S}^{(i)}_{\mathrm{cons}}$.
Then, the opinion configuration vector $y_i$ at the $i$-th round satisfies $y_i = S_i\dots S_1y_0=\sigma_w\mathbbm{1}$ for some $w\in V$,
    where $\sigma_w\in \Sigma$ is the opinion
    that $w\in V$ initially holds.
Since this relation holds regardless of the
    labels $\Sigma$ of opinions, we have
    $S_i\dots S_1 e_w = \mathbbm{1}$
    and $S_i\dots S_1 e_v = \mathbf{0}$ for $v\neq w$, where $\mathbf{0}$ denotes the all-zero vector.
Therefore,
\begin{align*}
    \mathbbm{1}^\top S_iS_{i-1}\dots S_1 e_v
    &= \begin{cases}
    |V| & \text{if $v=w$},\\
    0 & \text{otherwise}
    \end{cases}
\end{align*}
and thus $\mathbbm{1}^\top \prod_{t=1}^i S_{i-t+1} = |V| e_w^\top$.
In other words, $(S'_1,\dots,S'_i)\defeq (S_i,\dots,S_1)\in \mathcal{S}^{(i)}_{\mathrm{coal}}$.
The converse direction (i.e., $(S'_1,\dots,S'_i)\in\mathcal{S}^{(i)}_{\mathrm{coal}}$ implies $(S'_i,\dots,S'_1)\in\mathcal{S}^{(i)}_{\mathrm{cons}}$) can be checked similarly:
If $(S'_1,\dots,S'_i)\in\mathcal{S}^{(i)}_{\mathrm{coal}}$, then $\mathbbm{1}^\top S'_1\dots S'_i=|V|e_w^\top$ for some $w\in V$.
Then, for $(S_1,\dots,S_i)\defeq (S'_i,\dots,S'_1)$, we have
$\mathbbm{1}^\top S_i\cdots S_1e_w=\mathbbm{1}^\top S'_1\cdots S'_i e_w = |V|$.
Since $\mathbf{0}\leq S_i\cdots S_1 e_w \leq S_i\cdots S_1\mathbbm{1}=\mathbbm{1}$ (here, we write $(x_1,\dots,x_n)\leq (y_1,\dots,y_n)$ if $x_i\leq y_i$ for all $i\in[n]$), we have $S_i\cdots S_1e_w = \mathbbm{1}$.
This implies $(S_1,\dots,S_i)\in\mathcal{S}^{(i)}_{\mathrm{cons}}$.

The mapping $\phi\colon(S_1,\dots,S_i)\mapsto (S_i,\dots,S_1)$ is a bijection between $\mathcal{S}^{(i)}_{\mathrm{cons}}$ and $\mathcal{S}^{(i)}_{\mathrm{coal}}$ preserving the product measure $\prod_{t=1}^i \mu_P(S_t)$.
This implies that
    \cref{eq:taucons,eq:taucoal} are equal,
completing the proof of \cref{prop:duality_dynamic}.
\end{proof}
\subsection{Consensus time}
\begin{proof}[Proof of \cref{thm:consensus_time}]
If $\Pseq$ is irreducible, lazy, and reversible with respect to $\pi\in(0,1]^V$, so does $\mathcal{Q}^{(i)}$ for all $i\geq 0$.
Therefore, from \cref{thm:coalescing_time_bound},
 $\E[\taucoal(\mathcal{Q}^{(i)}]\leq T$ for all $i$,
    where $T=C\cdot\thitmax(\Pseq)$ for some absolute constant $C>0$.
Then, from \cref{prop:duality_dynamic}, we have
\begin{align*}
    \Pr[\taucons(\Pseq)\geq 2T] = \Pr[\taucoal(\mathcal{Q}^{(2T)})\geq 2T] \leq \frac{\E[\taucoal(\mathcal{Q}^{(2T)})]}{2T} \leq \frac{1}{2}. 
\end{align*}
Here, the initial opinion configuration is the worst one that all vertices have $n$ distinct opinions.
Therefore, for any fixed $t\geq 0$,
    $\Pr[\taucons((P_{t+s})_{s\geq 0})>2T]\leq 1/2$ holds for any initial opinion configuration.
From \cref{cor:key_cor} we have $\E[\taucons(\Pseq)]\leq 4T=4C\thitmax(\Pseq)=O(\thitmax(\Pseq))$.
\end{proof}

\subsection{Winning probability}
We prove \cref{prop:winning_probability_dynamic}.
Our proof is based on the 
    voting martingale argument
    that was used to obtain
    the winning probability result
    for the static graph setting (cf.~\cite{HP01,CR16}).
We just verify that the argument
    works for our dynamic graph setting.
For completeness, we write the proof
in this subsection.

\begin{proof}[Proof of \cref{prop:winning_probability_dynamic}.]
We first consider the special case of $\Sigma=\{0,1\}$ and then go on to the general case $\Sigma\subseteq\{0,\dots,n-1\}$.

\paragraph*{The case of $\Sigma=\{0,1\}$.}
Let $(Y_t)_{t\geq 0}$
    be the pull voting according to
    $\Pseq=(P_t)_{t\geq 1}$
    with a time-homogeneous
    stationary distribution $\pi$.
Note that $Y_t\in\{0,1\}^V$.
In this proof, we promise that
    $Y_t$ is an $n\times 1$ vector
    and $\pi\in[0,1]^V$ is a $1\times n$ vector.
Let $\pi(Y_t)= \sum_{u\in V}\pi(u)Y_t(u)$.
We first claim that $(\pi(Y_t))_{t\geq 0}$
is a martingale with respect to $Y_t$.
To see this, observe
\begin{align*}
    \E[\pi(Y_{t+1})|Y_t] &= \sum_{u\in V}\pi(u)\Pr[Y_{t+1}(u)=1 | Y_t] 
    = \sum_{u\in V}\pi(u)\sum_{w\in V}(P_t)(u,w)Y_t(w) \\
    &= \sum_{w\in V}(\pi P_t)(w) Y_t(w) 
    = \sum_{w\in V} \pi(w)Y_t(w) 
    = \pi(Y_t).
\end{align*}
Since $\pi(Y_t)$ are bounded,
we can apply the Optimal Stopping Theorem
and obtain $\E[\pi(Y_{\taucons})]=\pi(Y_0)$.
Note that, since $Y_{\taucons}$ is either all-zero or all-one, we have
\begin{align*}
    \E[\pi(Y_{\taucons})] &= \pi(\mathbbm{1})\Pr[Y_{\taucons}=\mathbbm{1}] = \Pr[Y_{\taucons}=\mathbbm{1}].
\end{align*}
In other words,
the probability that the opinion $1$ wins is equal to $\pi(Y_0)=\sum_{v\in V:Y_0(v)=1}\pi(v)$.

\paragraph*{General case.}
We reduce the general case to the binary opinion case by
    regarding $\sigma=1$ for a fixed opinion $\sigma\in\Sigma$ and all other opinions are zero.
Then we obtain the claim from the argument for the binary opinion case.
\end{proof}

\subsection{Exponential consensus time} \label{sec:exponential_lower_bound_pull}
In this subsection, we prove \cref{prop:exponential_lower_bound_pull}.
Consider the graph $G=(U\cup W,E)$ and permutation $\eta\colon V(G)\to V(G)$ defined in \cref{sec:meeting_time_lower_bound}.
Suppose vertices in $U$ have opinion $0$ and that in $W$ have opinion $1$ initially.
We claim that the sequence $(G_t)_{t\geq 1}$ given by $G_1=G$ and $G_{t}=\eta(G_{t-1})$ with the initial opinion configuration above has the exponential consensus time.

By the monotonicity of the pull voting,
we consider the following setting.
Suppose only vertices in $U=\{u_0,\dots,u_{m-1}\}$ perform the pull voting and vertices in $W$ always have opinion $1$.
Let $\tau$ be the time to reach the opinion configuration where all vertices in $U$ have opinion $1$.
It suffices to prove $\E[\tau]=2^{\Omega(m)}$, where $|U|=m$.

Since the graph dynamics is given by
    the iteration of applying the permutation $\eta$, it is convenient to consider
    the following equivalent process:
At the $t$-th round, vertices perform the one-round pull voting and then
opinions are shuffled according to $\eta$.
That is, if $Y_{t-1}\in\{0,1\}^U$ denote the
    opinion configuration at the beginning of the $t$-th round,
    we first update $Y_{t-1}$ by the pull voting to obtain $Z_t\in\{0,1\}^U$ and then set $Y_{t}(\eta(u))=Z_t(u)$ for every $u\in U$.
Note that the permutation $\eta$
    defined in \cref{sec:meeting_time_lower_bound}
    satisfies $\eta(U)=U$.
We consider the sequence $(Y_t)_{t\geq 0}$ described above, where $Y_0$ is the all-zero vector.
Note that the process agrees with the opinion $1$ at the last of the $t$-th round if and only if $Y_t(u)=1$ for all $u\in U$.

Let $\tau'=\inf\{t\geq 0\colon Y_t(0)=1\}$.
Then, for any $i=1,\dots,m-1$, the vertex $u_i$ must choose the self-loop in the pull voting procedure at the $(\tau'-i)$-th round to keep $u_i$'s opinion $1$
(otherwise, $u_i$ selects $u_0$, who has opinion $0$).
This happens with probability $2^{-m+1}$
and therefore we have $\E[\tau]\geq \E[\tau']=2^{\Omega(m)}$.
This completes the proof of \cref{prop:exponential_lower_bound_pull}.

\section{Metropolis walk on edge-Markovian graph}
\label{sec:edge_Markovian}
In this section,
    we prove \cref{thm:LMRW_EM}.
Let $V$ be a vertex set with $|V|=n$ and $p,q\in [0,1]$ be two parameters.
Let $\left((Y_t(e))_{t\geq 0}\right)_{e\in \binom{V}{2}}$ be $\binom{n}{2}$ independent Markov chains, 
where each $(Y_t(e))_{t\geq 0}$ is the Markov chain\footnote{Formally, $\Pr[Y_t(e)=b_t|Y_0(e)=b_0,\ldots,Y_{t-1}(e)=b_{t-1}]=\Pr[Y_t(e)=b_t|Y_{t-1}(e)=b_{t-1}]=M(b_{t-1},b_t)$ holds for all $t\geq 1$, $(b_0,\ldots b_t)\in \{0,1\}^{t+1}$ and $e\in \binom{V}{2}$.} with the state space $\{0, 1\}$ and the transition matrix $M=M_{p,q}\defeq \left(\begin{array}{cc}1-p & p \\ q & 1-q\end{array}\right)$. 
Henceforth, write $Y_t=(Y_t(e))_{e\in \binom{V}{2}}$ for convenience.
The edge-Markovian graph $\mathcal{G}(n,p,q)$ is a sequence of random graphs $(G_t)_{t\geq 0}=((V,E_t))_{t\geq 0}$, where $E_t\defeq \left\{e\in \binom{V}{2} :Y_t(e)=1\right\}$ for all $t\geq 0$.
We show the following lemma in this section.

\begin{lemma}
\label{lem:EMexpander}
Suppose that $0<p+q\leq 1$ and $\frac{p}{p+q}\geq 32(c+1)\frac{\log n}{n}$ hold for an arbitrary $c>0$.
Let $\mathcal{G}(n,p,q)=(G_t)_{t\geq 0}$ be the edge-Markovian graph.
Let $I=I(p,q)\defeq \left \lceil \frac{\max\{1,\log(q/p)\}}{p+q}\right \rceil$ and $J\leq n^{c}/2$. 
For any $\ell\geq 0$ and $0\leq i\leq \ell$, let $S(\ell,i)\defeq\left(\frac{\ell(\ell-1)}{2}+i\right)(I+J)=\left(\sum_{j=1}^{\ell-1} j +i\right)(I+J)$.
%
Let $C=8192$.
Then, for any $Y_0=b\in \{0,1\}^{\binom{V}{2}}$, 
\begin{align*}
\Pr\left[
\bigwedge_{\ell=1}^\infty\bigvee_{i=1}^{\ell}
\left\{\trelmax\left(\left(\LMWP(G_t)\right)_{t=S(\ell,i-1)+I}^{S(\ell,i)}\right)\leq C\right\}\right]\geq 1-\frac{1}{n}.
\end{align*}
\end{lemma}
\begin{proof}
For $S\subseteq V$, let $Q_t(S)=\sum_{u\in S}\sum_{v\notin S}\pi(u)P_t(u,v)$ and $\pi(S)=\sum_{v\in S}\pi(v)$.
Let $C_t(S)=\sum_{u\in S}\sum_{v\notin S}Y_{t}(\{u,v\})$.
From \cref{def:Metropolis_walk}, we have
\begin{align}
    \frac{Q_t(S)}{\pi(S)}
    &=\frac{n}{|S|}\sum_{u\in S}\sum_{v\notin S}\frac{1}{n}\frac{Y_t(\{u,v\})}{2\max\{\deg(G_t,u),\deg(G_t,v)\}}
    \geq \frac{C_t(S)}{2|S|d_{\max}(G_t)}.
    \label{eq:cheegerMW}
\end{align}
Now, it is easy to check that 
\begin{align*}
M^t= \frac{1}{p+q}\left[\left(\begin{array}{cc}q & p \\ q & p\end{array}\right)+(1-p-q)^t\left(\begin{array}{cc}p & -p \\ -q & q\end{array}\right)\right]
\end{align*}
holds for any $t\geq 1$.
Hence, for any $e\in \binom{V}{2}$, $s\geq 0$, $I'\geq I=\left \lceil \frac{\max\{1,\log(q/p)\}}{p+q}\right \rceil$ and $b\in \{0,1\}$, 
$\Pr[Y_{s+I'}(e)=1|Y_s(e)=b]=M^{I'}(b,1)$ and 
\begin{align}
\frac{p}{2(p+q)}
\leq \frac{p}{p+q}\left(1-(1-p-q)^{I'}\right)
\leq M^{I'}(b,1)
\leq \frac{p}{p+q}\left(1+(1-p-q)^{I'}\frac{q}{p}\right)\leq \frac{2p}{p+q}.
\label{eq:independence}
\end{align}
%
Using \cref{eq:independence} and the Chernoff inequality (\cref{lem:Chernoff}), 
for any  $s\geq 0$, $I'\geq I$ and $b\in \{0,1\}^{\binom{V}{2}}$, we have
\begin{align}
    &\Pr\left[\bigvee_{S\subseteq V:1\leq |S|\leq n/2}\left\{C_{s+I'}(S)\leq \frac{1}{2}\frac{|S|np}{4(p+q)}\right\}\middle|Y_s=b\right]\nonumber \\
    &\leq \sum_{S\subseteq V:1\leq |S|\leq n/2}\exp\left(-\frac{|S|np}{32(p+q)}\right)
    \leq \sum_{S\subseteq V:1\leq |S|\leq n/2}\exp\left(-(c+1)|S|\log n\right)\nonumber \\
    &=\sum_{1\leq x\leq n/2}\binom{n}{x}\frac{1}{n^{(c+1)x}}
    \leq \frac{1}{n^{c}-1}. \label{eq:probem_1}
\end{align}
Note that $\E\left[C_{s+I'}(S)\middle|Y_s=b\right]\geq |S|(n-|S|)\frac{p}{2(p+q)}\geq \frac{|S|np}{4(p+q)}$ for any $1\leq |S|\leq n/2$.
Furthermore, since $\E\left[\deg(G_{s+I'},v)\middle|Y_s=b\right]\leq n\frac{2p}{p+q}$, we have
\begin{align}
    &\Pr\left[\bigvee_{v\in V}\left\{\deg(G_t,v)\geq 2 n\frac{2p}{p+q} \right\}\middle|Y_s=b\right]
    \leq \sum_{v\in V}\exp\left(-\frac{2np}{3(p+q)}\right)\leq 
    \frac{1}{n^{64(c+1)/3-1}}.
    \label{eq:probem_2}
\end{align}
Combining \cref{eq:cheegerMW,eq:probem_1,eq:probem_2}, it holds with probability at least $1-\frac{2}{n^{c}}$ that
\begin{align*}
\min_{S\subseteq V:0<\pi(S)\leq 1/2}\frac{Q_{s+I'}(S)}{\pi(S)}
&
    \geq \frac{|S|np}{8(p+q)}\frac{p+q}{2|S|4np}=\frac{1}{64}.
\end{align*}
Hence, applying \cref{lem:Cheeger} yields the following: 
For any $s\geq 0$, $I'\geq I$ and $b\in \{0,1\}^{\binom{V}{2}}$,
\begin{align}
\Pr\left[\trel\left(\LMWP(G_{s+I'})\right)\leq C\middle| Y_s=b\right]
\geq 1-\frac{2}{n^c}.
\label{eq:Markovian-1}
\end{align}
Here, $C=2\cdotp64^2=8192$.
\begin{figure}[t]
    \centering
    \includegraphics[width=\hsize]{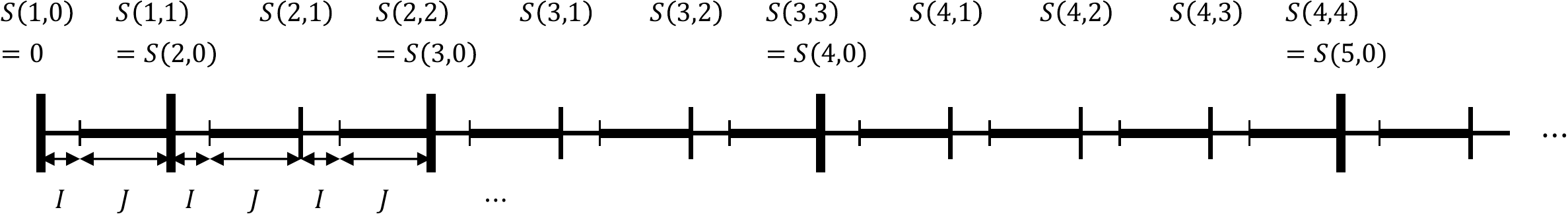}
    \caption{Parameters defined in \cref{lem:EMexpander}.}
    \label{fig:edgemarkovian}
\end{figure}
From \cref{eq:Markovian-1} and the union bound, 
\begin{align}
\Pr\left[\trelmax\left((\LMWP(G_t))_{t=s+I}^{s+I+J}\right)>C\middle| Y_s=b\right]
=\Pr\left[\bigvee_{t=s+I}^{s+I+J}\left\{\trel\left(\LMWP(G_t)\right)>C\right\}\middle| Y_s=b\right]
\leq \frac{2J}{n^c}\leq \frac{1}{n}
\label{eq:Markovian-2}
\end{align}
holds for any $s\geq 0$.
%
Let $\mathcal{E}_i^{(\ell)}$ be the event that $\trelmax\left((\LMWP(G_t))_{t=S(\ell,i-1)+I}^{S(\ell,i)}\right)>C$.
Fix $Y_0=b$.
From \cref{eq:Markovian-2}, we have
\begin{align*}
\Pr\left[ \bigwedge_{i=1}^{\ell} \mathcal{E}_i^{(\ell)} \right]
&=\sum_{b'\in \{0,1\}^{\binom{V}{2}}}\Pr\left[ \mathcal{E}_\ell^{(\ell)}\middle| \bigwedge_{i=1}^{\ell-1} \mathcal{E}_i^{(\ell)}, Y_{S(\ell,\ell-1)}=b'\right]\Pr\left[ \bigwedge_{i=1}^{\ell-1} \mathcal{E}_i^{(\ell)},Y_{S(\ell,\ell-1)}=b'\right]\\
&\leq \frac{1}{n}\Pr\left[ \bigwedge_{i=1}^{\ell-1} \mathcal{E}_i^{(\ell)}\right]
\leq \cdots \leq \frac{1}{n^\ell}.
\end{align*}
Hence, we obtain
\begin{align*}
\Pr\left[ \bigwedge_{\ell=1}^\infty \bigvee_{i=1}^\ell \overline{\mathcal{E}_\ell^{(i)}} \right]
&=1-\Pr\left[ \bigvee_{\ell=1}^\infty \bigwedge_{i=1}^\ell \mathcal{E}_\ell^{(i)} \right]
\geq 1-\sum_{\ell=1}^\infty\Pr\left[  \bigwedge_{i=1}^\ell \mathcal{E}_\ell^{(i)} \right]
\geq 1-\sum_{\ell=1}^\infty\frac{1}{n^\ell}.
\end{align*}
\end{proof}

\begin{proof}[Proof of the hitting time bound in \cref{thm:LMRW_EM}]
Let $J=\lceil 8192(2\log n+\log 2)+20\cdotp 8192 n/k\rceil$.
From \cref{lem:EMexpander,lem:hitting_probability_bound}, 
it holds for any $\ell\geq 1$ that
$\Pr[\tauhit^{(k)}(\Pseq)>S(\ell+1,0)|\tauhit^{(k)}(\Pseq)>S(\ell,0)]\leq 0.9$.
Applying \cref{lem:ex_stopping} yields
$\E[\tauhit^{(k)}(\Pseq)]\leq \sum_{\ell=0}^\infty\ell(I+J)(0.9)^\ell\leq 100(I+J)$.
\end{proof}
\begin{proof}[Proof of the cover time bound in \cref{thm:LMRW_EM}]
Let $J=\lceil 8192(2\log n+\log 2)+20\cdotp 8192 n\log n/k\rceil$.
From \cref{lem:EMexpander,lem:cover_probability_bound}, it holds for any $\ell\geq 1$ that
$\Pr[\taucov^{(k)}>S(\ell+1,0)|\taucov^{(k)}>S(\ell,0)]\leq 0.9$.
Applying \cref{lem:ex_stopping}, 
$\E[\taucov^{(k)}(\Pseq)]\leq \sum_{\ell=0}^\infty\ell(I+J)(0.9)^\ell\leq 100(I+J)$ holds.
\end{proof}
\begin{proof}[Proof of the coalescing time bound in \cref{thm:LMRW_EM}]
Let $J=\lceil 8192 Cn\rceil$.
From \cref{lem:EMexpander,lem:coalcons}, it holds for any $\ell\geq 1$ that
$\Pr[\taucoal(\Pseq)>S(\ell+1,0)|\taucoal(\Pseq)>S(\ell,0)]\leq 1-10^{-5}$ for any $\ell$.
Applying \cref{lem:ex_stopping}, 
$\E[\taucoal(\Pseq)]\leq \sum_{\ell=0}^\infty\ell(I+J)(1-10^{-5})^\ell\leq C'(I+J)$ holds for some absolute constant $C'>0$.
\end{proof}

\section{Conclusion}
We obtain new bounds of the mixing, hitting, and cover times of the random walk according to the sequence of irreducible, reversible, and lazy transition matrices that have the same stationary distribution.
These bounds generalize previous works for a lazy simple random walk or a $d_{\max}$-lazy walk and improve them in various cases.
Furthermore, we obtain the first bounds of the hitting and cover times of multiple random walks and the coalescing time on dynamic graphs.
Additionally, we bound the consensus time of the pull-voting on dynamic graphs.
Our results strengthen the observation that
    time inhomogeneous Markov chains
    with an invariant stationary distribution
    behaves almost identically to
    a time-homogeneous chain.
Specifically, we prove that if all $P_t$ have the same stationary distribution,
then $\thit\left((P_t)_{t\geq 1}\right)\leq O\left(\max_{t\geq 1}\thit(P_t)\right)$ holds  (\cref{thm:main_hitting_cover}).
It is natural to ask for the same relation for other parameters.
For example, does $\tcov(\mathcal{P})\leq O\left(\max_{t\geq 1}\tcov(P_t)\right)$ hold?

Most previous works on 
time-inhomogeneous random walks have
    translated techniques from time-homogeneous chains into 
    time-inhomogeneous ones:
In particular, several known upper bounds (including ours) are based on \emph{spectral} arguments,
which essentially requires the time-homogeneity of the stationary distribution.
On the other hand, known lower bounds such as the Sisyphus wheel
    are based on
    some \emph{combinatorial} arguments.
To understand time-inhomogeneous random walks with time-varying stationary distributions,
    it might be important to
    interpolate the spectral and combinatorial arguments.
    The simple random walk on
    any static connected graph
    has an $O(n^3)$ cover time.
This research question might be a possible future direction of the research of time-inhomogeneous chains.

\bibliographystyle{abbrv}
\bibliography{ref}

\begin{thebibliography}{10}

\bibitem{ACD15}
M.~Abdullah, C.~Cooper, and M.~Draief.
\newblock Speeding up cover time of sparse graphs using local knowledge.
\newblock {\em In Proceedings of the International Workshop on Combinatorial
  Algorithms (IWOCA)}, 1:1--12, 2015.

\bibitem{AF02}
D.~J. Aldous and J.~A. Fill.
\newblock Reversible {M}arkov chains and random walks on graphs.
\newblock https://www.stat.berkeley.edu/users/aldous/RWG/book.html.

\bibitem{Aleliunas79}
R.~Aleliunas, R.~M. Karp, R.~J. Lipton, L.~Lov{\'a}sz, and C.~Rackoff.
\newblock Random walks, universal traversal sequences, and the complexity of
  maze problems.
\newblock {\em In Proceedings of 20th Annual Symposium on Foundations of
  Computer Science (FOCS)}, pages 218--223, 1979.

\bibitem{AAKKLT11}
N.~Alon, C.~Avin, M.~Kouck{\'y}, G.~Kozma, Z.~Lotker, and M.~Tuttle.
\newblock Many random walks are faster than one.
\newblock {\em Combinatorics, Probability and Computing}, 20(4):2623--2641,
  2011.

\bibitem{AKL08}
C.~Avin, M.~Kouck{\'y}, and Z.~Lotker.
\newblock How to explore a fast-changing world (cover time of a simple random
  walk on evolving graphs).
\newblock {\em In Proceedings of the 35th International Colloquium on Automata,
  Languages, and Programming (ICALP)}, pages 121--132, 2008.

\bibitem{AKL18}
C.~Avin, M.~Kouck{\'y}, and Z.~Lotker.
\newblock Cover time and mixing time of random walks on dynamic graphs.
\newblock {\em Random Structures \& Algorithms}, 52(4):576--596, 2018.

\bibitem{BCF11}
H.~Baumann, P.~Crescenzi, and P.~Fraigniaud.
\newblock Parsimonious flooding in dynamic graphs.
\newblock {\em Distributed Computing}, 24:31--44, 2011.

\bibitem{BGKM16}
P.~Berenbrink, G.~Giakkoupis, A.-M. Kermarrec, and F.~Mallmann-Trenn.
\newblock Bounds on the voter model in dynamic networks.
\newblock {\em In Proceedings of the 43rd International Colloquium on Automata,
  Languages, and Programming (ICALP)}, 2016.

\bibitem{BW90}
G.~Brightwell and P.~Winkler.
\newblock Maximum hitting time for random walks on graphs.
\newblock {\em Random Structures \& Algorithms}, 1(3):263--276, 1990.

\bibitem{BKRU94}
A.~Broder, A.~Karlin, P.~Raghavan, and E.~Upfal.
\newblock Trading space for time in undirected {$s$}-{$t$} connectivity.
\newblock {\em SIAM Journal on Computing}, 23(2):324--334, 1994.

\bibitem{CSZ20}
L.~Cai, T.~Sauerwald, and L.~Zanetti.
\newblock Random walks on randomly evolving graphs.
\newblock {\em In Proceedings of the 27th International Colloquium on
  Structural Information and Communication Complexity (SIROCCO)}, 2020.

\bibitem{CCDFPS16}
A.~Clementi, P.~Crescenzi, C.~Doerr, P.~Fraigniaud, F.~Pasquale, and
  R.~Silvestri.
\newblock Rumor spreading in random evolving graphs.
\newblock {\em Random Structures \& Algorithms}, 48(2):290--312, 2016.

\bibitem{CMMPS10}
A.~Clementi, C.~Macci, A.~Monti, F.~Pasquale, and R.~Silvestri.
\newblock Flooding time of edge-markovian evolving graphs.
\newblock {\em SIAM Journal on Discrete Mathematics}, 24(4):1694--1712, 2010.

\bibitem{Cooper11}
C.~Cooper.
\newblock Random walks, interacting particles, dynamic networks: {R}andomness
  can be helpful.
\newblock {\em In Proceeedings of the 18th International Colloquium on
  Structural Information and Communication Complexity (SIROCCO)}, pages 1--14,
  2011.

\bibitem{CEOR13}
C.~Cooper, R.~Els{\"a}sser, H.~Ono, and T.~Radzik.
\newblock Coalescing random walks and voting on connected graphs.
\newblock {\em SIAM Journal on Discrete Mathematics}, 27(4):1748--1758, 2013.

\bibitem{CF03}
C.~Cooper and A.~Frieze.
\newblock Crawling on simple models of web graphs.
\newblock {\em Internet Mathematics}, 1(1):57--90, 2003.

\bibitem{CR16}
C.~Cooper and N.~Rivera.
\newblock The linear voting model.
\newblock {\em In Proceedings of the 43rd International Colloquium on Automata,
  Languages, and Programming (ICALP)}, 55(144):1--12, 2016.

\bibitem{DF18}
R.~David and U.~Feige.
\newblock Random walks with the minimum degree local rule have {$O(n^2)$} cover
  time.
\newblock {\em SIAM Journal on Computing}, 47(3):755--768, 2018.

\bibitem{DR14}
O.~Denysyuk and L.~Rodrigues.
\newblock Random walks on evolving graphs with recurring topologies.
\newblock {\em In Proceedings of the 28th International Symposium on
  Distributed Computing (DISC)}, pages 333--345, 2014.

\bibitem{Doerr18}
B.~Doerr and F.~Neumann.
\newblock {\em Theory of evolutionary computation: Recent developments in
  discrete optimization}.
\newblock Springer International Publishing, 2020.

\bibitem{ER09}
K.~Efremenko and O.~Reingold.
\newblock How well do random walks parallelize?
\newblock {\em in Proceedings of the 13th International Workshop on
  Randomization and Approximation Techniques in Computer Science (RANDOM)},
  pages 476--489, 2009.

\bibitem{ES11}
R.~Els{\"a}sser and T.~Sauerwald.
\newblock Tight bounds for the cover time of multiple random walks.
\newblock {\em Theoretical Computer Science}, 412(24):2623--2641, 2011.

\bibitem{Feige95up}
U.~Feige.
\newblock A tight upper bound on the cover time for random walks on graphs.
\newblock {\em Random Structures \& Algorithms}, 6(1):51--54, 1995.

\bibitem{Fill91}
J.~A. Fill.
\newblock Eigenvalue bounds on convergence to stationarity for nonreversible
  {M}arkov chains, with an application to the exclusion process.
\newblock {\em The Annals of Applied Probability}, pages 62--87, 1991.

\bibitem{Gri75}
D.~Griffeath.
\newblock Uniform coupling of non-homogeneous markov chains.
\newblock {\em Journal of Applied Probability}, 12(4):753--762, 1975.

\bibitem{HP01}
Y.~Hassin and D.~Peleg.
\newblock Distributed probabilistic polling and applications to proportionate
  agreement.
\newblock {\em Information and Computation}, 171(2):248--268, 2001.

\bibitem{HJ12}
R.~Horn and C.~Johnson.
\newblock {\em Matrix Analysis: Second Edition}.
\newblock Campridge University Press, 2012.

\bibitem{Ikeda03}
S.~Ikeda, I.~Kubo, N.~Okumoto, and M.~Yamashita.
\newblock Impact of local topological information on random walks on finite
  graphs.
\newblock {\em In Proceedings of the 30th International Colloquium on Automata,
  Languages and Programming (ICALP)}, pages 1054--1067, 2003.

\bibitem{Ikeda09}
S.~Ikeda, I.~Kubo, and M.~Yamashita.
\newblock The hitting and cover times of random walks on finite graphs using
  local degree information.
\newblock {\em Theoretical Computer Science}, 410(1):94--100, 2009.

\bibitem{KLNS89}
J.~Kahn, N.~Linial, N.~Nisan, and M.~Saks.
\newblock On the cover time of random walks on graphs.
\newblock {\em Journal of Theoretical Probability}, 2:121--128, 1989.

\bibitem{KMS19}
V.~Kanade, F.~Mallmann-Trenn, and T.~Sauerwald.
\newblock On coalescence time in graphs: {W}hen is coalescing as fast as
  meeting?
\newblock {\em In Proceedings of the 30th Annual ACM-SIAM Symposium on Discrete
  Algorithms (SODA)}, pages 956--965, 2019.

\bibitem{KSS21}
S.~Kijima, N.~Shimizu, and T.~Shiraga.
\newblock How many vertices does a random walk miss in a network with
  moderately increasing the number of vertices?
\newblock {\em In Proceedings of the 2021 ACM-SIAM Symposium on Discrete
  Algorithms (SODA)}, pages 106--122, 2021.

\bibitem{LMS18}
I.~Lamprou, R.~Martin, and P.~Spirakis.
\newblock Cover time in edge-uniform stochastically-evolving graphs.
\newblock {\em Algorithms}, 11(10):149, 2018.

\bibitem{LP17}
D.~A. Levin and Y.~Peres.
\newblock {\em {M}arkov Chain and Mixing Times: Second Edition}.
\newblock The American Mathematical Society, 2017.

\bibitem{Lovasz93}
L.~Lov{\'a}sz.
\newblock Random walks on graphs: A survey.
\newblock {\em Combinatorics, Paul Erd\H{o}s is Eighty}, 2:1--46, 1993.

\bibitem{Matthews88}
P.~Matthews.
\newblock Covering problems for {M}arkov chains.
\newblock {\em The Annals of Probability}, 16(3):1215--1228, 1988.

\bibitem{MPS04}
E.~Mossel, Y.~Peres, and A.~Sinclair.
\newblock Shuffling by semi-random transpositions.
\newblock {\em In Proceedings of the 45th Annual IEEE Symposium on Foundations
  of Computer Science (FOCS)}, pages 572--581, 2004.

\bibitem{NOSY10}
Y.~Nonaka, H.~Ono, K.~Sadakane, and M.~Yamashita.
\newblock The hitting and cover times of {M}etropolis walks.
\newblock {\em Theoretical Computer Science}, 411(16--18):1889--1894, 2010.

\bibitem{Oliveira12}
R.~I. Oliveira.
\newblock On the coalescing time of reversible random walks.
\newblock {\em Transactions of the American Mathematical Society},
  364:2109--2128, 2012.

\bibitem{OP19}
R.~I. Oliveira and Y.~Peres.
\newblock Random walks on graphs: {N}ew bounds on hitting, meeting, coalescing
  and returning.
\newblock {\em In Proceedings of the 16th Workshop on Analytic Algorithmics and
  Combinatorics (ANALCO)}, pages 119--126, 2019.

\bibitem{OT11}
A.~Olshevsky and J.~N. Tsitsiklis.
\newblock Degree fluctuations and the convergence time of consensus algorithms.
\newblock {\em In Proceedings of the 50th IEEE Conference on Decision and
  Control}, pages 6602--6607, 2011.

\bibitem{RSS21}
N.~Rivera, T.~Sauerwald, and J.~Sylvester.
\newblock Multiple random walks on graphs: Mixing few to cover many.
\newblock {\em In Proceedings of the 48th International Colloquium on Automata,
  Languages, and Programming (ICALP)}, pages 107:1--107:16, 2021.

\bibitem{SZ07}
L.~Saloff-Coste and J.~Z{\'u}{\~n}iga.
\newblock Convergence of some time inhomogeneous markov chains via spectral
  techniques.
\newblock {\em Stochastic Processes and their Applications}, 117(8):961--979,
  2007.

\bibitem{SZ09}
L.~Saloff-Coste and J.~Z{\'u}{\~n}iga.
\newblock Merging for time inhomogeneous finite {M}arkov chains, part {I}:
  {S}ingular values and stability.
\newblock {\em Electronic Journal of Probability}, 14(49):1456--1494, 2009.

\bibitem{SZ11}
L.~Saloff-Coste and J.~Z{\'u}{\~n}iga.
\newblock Merging for inhomogeneous finite {M}arkov chains, part {II}: Nash and
  log-{S}obolev inequalities.
\newblock {\em Annals of Probability}, 39(3):1161--1203, 2011.

\bibitem{SZ19}
T.~Sauerwald and L.~Zanetti.
\newblock Random walks on dynamic graphs: Mixing times, hitting times, and
  return probabilities.
\newblock {\em In Proceedings of the 46th International Colloqium on Automata,
  Languages, and Programming (ICALP)}, pages 93:1--93:15, 2019.

\end{thebibliography}
\appendix

\section{Tools for key lemmas}
In this section, we introduce technical tools for \cref{lem:mix_treldecay,lem:dirichlethit,lem:MTL_formal,lem:HTL_formal}.
The first one is concerned with the spectral radius $\rho(D_w PD_w)$ of the substochastic matrix $D_wPD_w$ (see \cref{sec:proof_outline} for the definition of $D_w$).
It is known that $\rho(D_wPD_w)\leq 1-1/\thit(P)$ for any irreducible $D_wPD_w$ (Section 3.6.5 of \cite{AF02}).
For completeness, we show this under the assumption of irreducibility and reversibility of $P$ as follows.
\begin{lemma}
\label{lem:hiteigen}
Let $P\in [0,1]^{V\times V}$ be an irreducible and reversible transition matrix.
Then, for any $w\in V$, 
\begin{align*}
\rho\left(D_w P D_w\right) \leq 1-\frac{1}{\thit(P)}.
\end{align*}
\end{lemma}
\begin{proof}
Define $P_w\in [0,1]^{V\setminus \{w\}\times V\setminus \{w\}}$ by $P_w(u,v)=P(u,v)$ for any $u,v\in V\setminus \{w\}$.
The Perron--Frobenius theorem
implies that
$\lambda=\rho(P_w)$ is an eigenvalue of $P_w$ and there is a nonnegative nonzero eigenvector $g\in \mathbbm{R}^{V\setminus \{w\}}$ satisfying $P_wg=\lambda g$ (see, e.g., Theorem 8.3.1 in \cite{HJ12}).
Write $Q_w=D_wPD_w$ for convenience.
Define $h\in \mathbbm{R}^V$ by $h(w)=0$ and $h(v)=\frac{g(v)\pi(v)}{Z}$ for any $v\in V\setminus \{w\}$, where $Z=\sum_{v\in V\setminus \{w\}}g(v)\pi(v)$.
Then, $h$ is a probability vector.
Furthermore, 
\begin{align*}
(hQ_w)(v)
&=\sum_{u\in V}h(u)Q_w(u,v)
=\sum_{u\in V\setminus \{w\}}\frac{g(u)\pi(u)}{Z}P_w(u,v)\\
&=\sum_{u\in V\setminus \{w\}}\frac{g(u)\pi(v)}{Z}P_w(v,u)
=\frac{\pi(v)}{Z}(P_wg)(v)
=\frac{\pi(v)}{Z}\lambda g(v)
=\lambda h(v)
\end{align*}
holds for any $v\in V\setminus \{w\}$. Since $(hQ_w)(w)=0=\lambda h(w)$, we have $hQ_w=\lambda h$.
Hence, $hQ_w^t=\lambda^th$ holds for any $t\geq 1$.
This implies that
\begin{align*}
\Pr_{h}\left[\tau_w>t\right]
&=\Pr_{h}\left[\bigwedge_{i=0}^t\{X_i\neq w\}\right]
=\sum_{v\in V}h(v) \sum_{u\in W}Q_w^t(v,u)
=\lambda^t\sum_{u\in V}h(u)=\lambda^t
\end{align*}
holds for any $t\geq 1$.
Since $P$ is irreducible, there is a $t^* \geq 1$ such that $\Pr_{h}\left[\tau_w>t^*\right]<1$. 
Hence, $\lambda<1$ and we have
\begin{align*}
\E_{h}\left[\tau_w\right]
&=\sum_{t=0}^\infty\Pr_{h}\left[\tau_w>t\right]
 =\frac{1}{1-\lambda}.
\end{align*}
Note that $\Pr_{h}\left[\tau_w>0\right]=1$ holds from $h(w)=0$.
Thus, $\rho(D_wPD_w)=\lambda=1-\frac{1}{\E_{h}\left[\tau_w\right]}\leq 1-\frac{1}{\thit(P)}$ holds and we obtain the claim.
Note that we have $\E_{h}\left[\tau_w\right]=\sum_{v\in V}h(v)\E_{v}\left[\tau_w\right]\leq \sum_{v\in V}h(v)\thit(P)=\thit(P)$.
\end{proof}

The following lemmas are already known in
    the literature.
We put the proofs of them for completeness.
\begin{lemma}
\label{lem:MatrixEigen}
Let $M\in \mathbbm{R}^{V\times V}$ be a matrix and $\nu\in \mathbbm{R}^V_{>0}$ be a positive vector.
Suppose that $\nu(u)M(u,v)=\nu(v)M(v,u)$ holds for all $u,v\in V$.
Then, 
\begin{align*}
\langle M f, f\rangle_\nu\leq \rho(M)\langle f, f\rangle_\nu \hspace{1em}\textrm{and} \hspace{1em} 
\|Mf\|_{2,\nu}\leq \rho(M)\|f\|_{2,\nu}
\end{align*}
hold for any $f\in \mathbbm{R}^V$.
Furthermore, if $M$ is a transition matrix, 
\begin{align*}
\langle M f, f\rangle_\nu\leq \lambda_\star(M)\langle f, f\rangle_\nu \hspace{1em}\textrm{and} \hspace{1em} 
\|Mf\|_{2,\nu}\leq \lambda_\star(M)\|f\|_{2,\nu}
\end{align*}
hold for any $f\in \mathbbm{R}^V$ satisfying $\langle f,\mathbbm{1}\rangle_\nu=0$.
\end{lemma}
\begin{proof}
From the assumption, $\langle Mf,g\rangle_\nu=\langle f,Mg\rangle_\nu$ holds for any $f,g\in \mathbbm{R}^V$.
Hence, from the spectral theorem, the inner product space $(\mathbbm{R}^V,\langle \cdot,\cdot \rangle_{\nu})$ has an orthonormal basis of real-valued eigenvectors $\{\psi_i\}_{i=1}^{|V|}$ corresponding to real eigenvalues $\{\lambda_i(M)\}_{i=1}^{|V|}$ (see, e.g., Lemma 12.1 in \cite{LP17}).
In other words, for any $i,j$ and $f\in \mathbbm{R}^V$, we have $M\psi_i=\lambda_i(M)\psi_i$, $\langle \psi_i,\psi_j\rangle_\nu=\mathbbm{1}_{i=j}$, and $f=\sum_{i=1}^{|V|}\langle f,\psi_i\rangle_\nu\psi_i$.
Without loss of generality, assume $\lambda_1(M)\geq \lambda_2(M)\geq \cdots \lambda_{|V|}(M)$.
For any $f\in \mathbbm{R}^V$, we have
\begin{align}
&\|f\|_{2,\nu}^2
=\langle f,f\rangle_\nu
=\left \langle \sum_{i=1}^{|V|}\langle f,\psi_i\rangle_\nu\psi_i,f \right \rangle_\nu
=\sum_{i=1}^{|V|}\langle f,\psi_i\rangle_\nu^2, \label{eq:fnorm}\\
&\|Mf\|_{2,\nu}^2
=\sum_{i=1}^{|V|}\langle Mf,\psi_i\rangle_\nu^2
=\sum_{i=1}^{|V|}\langle f,M\psi_i\rangle_\nu^2
=\sum_{i=1}^{|V|}\lambda_i(M)^2\langle f,  \psi_i\rangle_\nu^2, \label{eq:Mfnorm}\\
&\langle Mf,f \rangle_\nu
=\left \langle \sum_{i=1}^{|V|}\langle Mf,\psi_i\rangle_\nu\psi_i,f \right \rangle_\nu
=\sum_{i=1}^{|V|}\langle f,M\psi_i\rangle_\nu\left \langle \psi_i,f \right \rangle_\nu
=\sum_{i=1}^{|V|}\lambda_i(M) \langle f,\psi_i  \rangle_\nu^2. \label{eq:Mprod}
\end{align}
Combining \cref{eq:Mprod,eq:fnorm}, $\langle Mf,f \rangle_\nu\leq \rho(M)\langle f,f\rangle_\nu$ holds.
Combining \cref{eq:Mfnorm,eq:fnorm}, $\|Mf\|_{2,\nu}^2\leq \rho(M)^2\|f\|_{2,\nu}^2$ holds.
If $M$ is a transition matrix, we have $\lambda_1(M)=1$ and $\psi_1=\mathbbm{1}$. Furthermore, $|\lambda_i(M)|\leq 1$ holds for all $1\leq i\leq |V|$.
Combining \cref{eq:Mprod,eq:fnorm}, $\langle Mf,f \rangle_\nu\leq \lambda_1(M)\langle f,\psi_1\rangle_\nu+\lambda_\star(M)\langle f,f\rangle_\nu=\lambda_\star(M)\langle f,f\rangle_\nu$ holds.
Combining \cref{eq:Mfnorm,eq:fnorm}, $\|Mf\|_{2,\nu}^2\leq \lambda_1(M)^2\langle f,\psi_1\rangle_\nu^2+\lambda_\star(M)^2\langle f,f\rangle_\nu=\lambda_\star(M)^2\|f\|_{2,\nu}^2$ holds.
\end{proof}

\begin{lemma}[See, e.g., (12.8) of \cite{LP17}]
\label{lem:normmix}
Let $P\in [0,1]^{V\times V}$ be a transition matrix.
Suppose that $\pi(u)P(u,v)=\pi(v)P(v,u)$ holds for any $u,v\in V$ and some probability distribution $\pi\in (0,1]^V$.
Then, for any probability vector $\mu\in [0,1]^V$, 
\begin{align*}
\left\|\frac{\mu P}{\pi}-\mathbbm{1}\right\|_{2,\pi}^2\leq \lambda_\star(P)^2 \left\|\frac{\mu}{\pi}-\mathbbm{1}\right\|_{2,\pi}^2.
\end{align*}
\end{lemma}
\begin{proof}
Combining \cref{eq:reversibled2norm,lem:MatrixEigen}, we have
\begin{align*}
\left\|\frac{\mu P}{\pi}-\mathbbm{1}\right\|_{2,\pi}^2
&=\left\|P\left(\frac{\mu}{\pi}\right)-P\mathbbm{1}\right\|_{2,\pi}^2
=\left \|P\left(\frac{\mu}{\pi}-\mathbbm{1}\right)\right \|_{2,\pi}^2
\leq \lambda_\star(P)^2\left\|\frac{\mu}{\pi}-\mathbbm{1}\right\|_{2,\pi}^2.
\end{align*}
Note that we have
$
\left \langle \frac{\mu}{\pi}-\mathbbm{1},\mathbbm{1} \right\rangle_{\pi}
=\sum_{v\in V}\pi(v)\left(\frac{\mu(v)}{\pi(v)}-1\right)
=0
$.
\end{proof}
%
\begin{lemma}[See, e.g., Proposition 2.5 in \cite{Fill91}]
\label{lem:Mihail}
Let $P\in [0,1]^{V\times V}$ be a lazy transition matrix.
Suppose that $\pi(u)P(u,v)=\pi(v)P(v,u)$ holds for any $u,v\in V$ and some probability distribution $\pi\in (0,1]^V$.
Then for any $f\in \mathbbm{R}^V$,
\begin{align*}
\Var_\pi(Pf)
\leq \Var_\pi(f)-\mathcal{E}_{P,\pi}(f,f).
\end{align*}
\end{lemma}
\begin{proof}
It is straightforward to see that
\begin{align*}
\Var_\pi(Pf)
&=\langle Pf,Pf \rangle_{\pi}-\langle Pf,\mathbbm{1} \rangle_{\pi}^2
=\langle P^2f,f \rangle_{\pi}-\langle f,P\mathbbm{1} \rangle_{\pi}^2\\
&=\left(\langle f,f \rangle_{\pi}-\langle f,\mathbbm{1} \rangle_{\pi}^2\right) - \left(\langle f,f \rangle_{\pi}-\langle P^2f,f \rangle_{\pi}\right) 
=\Var_\pi(f)-\mathcal{E}_{P^2,\pi}(f,f)
\end{align*}
holds. 
From \cref{eq:fnorm,eq:Mprod}, $\mathcal{E}_{P,\pi}(f,f)=\langle f,f \rangle_{\pi}-\langle Pf,f \rangle_{\pi}=\sum_{i=2}^{|V|}(1-\lambda_i(P))\langle f,\psi_i\rangle_\pi^2$.
Hence, we have $\mathcal{E}_{P^2,\pi}(f,f)=\sum_{i=2}^{|V|}(1-\lambda_i(P)^2)\langle f,\psi_i \rangle_\pi^2\geq \sum_{i=2}^{|V|}(1-\lambda_i(P))\langle f,\psi_i \rangle_\pi^2=\mathcal{E}_{P,\pi}(f,f)$. Thus, we obtain the claim.
Note that all eigenvalues of $P$ are non-negative since $P$ is lazy.
\end{proof}
%

%
\begin{lemma}[See, e.g., Theorem 4.1 in \cite{OP19}]
\label{lem:normmeet}
Let $P\in [0,1]^{V\times V}$ be a lazy transition matrix.
Suppose that $\pi(u)P(u,v)=\pi(v)P(v,u)$ holds for any $u,v\in V$ and some probability distribution $\pi\in (0,1]^V$.
Then for any $x,y\in V$ and any $f\in \mathbbm{R}^V$, 
\begin{align*}
\left\|D_xPD_yf\right\|_{2,\pi}^2\leq \rho\left(D_xPD_x \right)\rho\left(D_yPD_y\right)\left\|f\right\|_{2,\pi}^2.
\end{align*}
\end{lemma}
\begin{proof}
From assumption,
the inner product space $(\mathbbm{R}^V,\langle\cdot,\cdot\rangle_{\pi})$ has an orthonormal basis of real-valued eigenvectors $\{\psi_i\}_{i=1}^{|V|}$ corresponding to real eigenvalues $\{\lambda_i(P)\}_{i=1}^{|V|}$.
This implies that, for all $u,v\in V$,  $P(v,u)=\pi(u)\sum_{i=1}^{|V|}\lambda_i(P)\psi_i(v)\psi_i(u)$ holds.
Let $\sqrt{P}\in [0,1]^{V\times V}$ be the positive semidefinite square root of $P$, i.e., $\sqrt{P}(v,u)=\pi(u)\sum_{i=1}^{|V|}\sqrt{\lambda_i(P)}\psi_i(v)\psi_i(u)$.
Note that all eigenvalues are nonnegative since $P$ is lazy.
It is easy to see that $(\sqrt{P})^2=P$ and $\pi(v)\sqrt{P}(v,u)=\pi(u)\sqrt{P}(u,v)$ holds for any $u,v\in V$.
Hence, we have
\begin{align*}
\pi(v)(D_w\sqrt{P})(v,u) 
=\pi(v)D_w(v,v)\sqrt{P}(v,u) 
=\pi(u)D_w(v,v)\sqrt{P}(u,v) 
=\pi(u)(\sqrt{P}D_w)(u,v),
\end{align*}
i.e., $\langle D_w\sqrt{P}f,g\rangle_\pi=\langle f,\sqrt{P}D_wg\rangle_\pi$ holds for any $f,g$.
This implies that both  $\|D_w\sqrt{P}f\|_{2,\pi}$
 and $\|\sqrt{P}D_w f\|_{2,\pi}$ are upper bounded by $\sqrt{\rho(D_wPD_w)}\|f\|_{2,\pi}$.
Consequently, 
\begin{align*}
\left\| D_xPD_y f\right\|_{2,\pi}^2
&=\left\| D_x\sqrt{P} \sqrt{P}D_y f \right\|_{2,\pi}^2 
\leq \rho\left(D_xPD_x\right)\left\|  \sqrt{P}D_y f \right\|_{2,\pi}^2 
&=\rho\left(D_xPD_x \right)\rho\left(D_yPD_y \right)\left\| f \right\|_{2,\pi}^2
\end{align*}
holds, and we obtain the claim.
\end{proof}

\section{Tools for expected stopping time}
Our upper bounds of the hitting, cover, and coalescing times rely on the following observation.
\begin{lemma}\label{lem:key_lemma}
Let $(Z_t)_{t\geq 0}$ be a sequence of
    random variables where $Z_t\in\mathcal{S}$
    for a finite state space $\mathcal{S}$.
For an event $\mathcal{E}\subseteq\mathcal{S}$,
    let $\tau(\mathcal{E})=\inf\{t\geq 0:Z_t\in\mathcal{E}\}$ be the stopping time.
Suppose there exist $T>0$ and $c>0$ such that,
    for any $t\geq 0$,
    \begin{align*}
        \Pr[\tau(\mathcal{E})\geq T+t \mid \tau(\mathcal{E})\geq t] \leq 1-c
    \end{align*}
    holds.
Then, $\E[\tau(\mathcal{E})]\leq \frac{T}{c}$.
\end{lemma}
\begin{proof}
From the assumption, for any $k\geq 0$, we have
\begin{align*}
    \Pr\left[\tau(\mathcal{E})\geq kT+t\right]
    &=
    \Pr\left[\tau(\mathcal{E})\geq kT+t \mid \tau(\mathcal{E})\geq (k-1)T+t\right]
    \Pr\left[\tau(\mathcal{E})\geq (k-1)T+t\right] \\
    &\leq (1-c)\Pr\left[\tau(\mathcal{E})\geq (k-1)T+t\right] \\
    &\dots\\
    &\leq (1-c)^k.
\end{align*}
Therefore, we have
\begin{align*}
    \E[\tau(\mathcal{E})] =
        \sum_{k=0}^\infty \sum_{t=0}^{T-1}
        \Pr[\tau(\mathcal{E})\geq kT+t]
        \leq T\sum_{k=0}^\infty (1-c)^k
        = \frac{T}{c}.
\end{align*}
\end{proof}
\begin{corollary}\label{cor:key_cor}
Let $(Z_t)_{t\geq 0}$ be a sequence of
    random variables where $Z_t\in\mathcal{S}$
    for a finite state space $\mathcal{S}$.
For an event $\mathcal{E}\subseteq\mathcal{S}$,
    let $\tau(\mathcal{E})=\inf\{t\geq 0:Z_t\in\mathcal{E}\}$ be the stopping time.
Suppose there exist $T>0$ and $c>0$ such that,
    for any $t\geq 0$ and $z\in\mathcal{S}$,
    \begin{align*}
        \Pr[\tau(\mathcal{E})\geq T+t \mid Z_t=z] \leq 1-c
    \end{align*}
    holds.
Then, $\E[\tau(\mathcal{E})]\leq \frac{T}{c}$.
\end{corollary}
\begin{proof}
Note that
\begin{align*}
    \Pr[\tau(\mathcal{E})\geq T+t \mid \tau(\mathcal{E})\geq t]
    \leq \Pr[\tau(\mathcal{E})\geq T+t \mid Z_t]
    \leq 1-c.
\end{align*}
\end{proof}
To obtain upper bounds for hitting, cover, and coalescing times, it suffices to
prove that the corresponding stopping time
    satisfies the condition of \cref{cor:key_cor}.
    
Finally, we introduce the following lemma, which we use in the proof of the edge-Markovian graph (\cref{sec:edge_Markovian}).
\begin{lemma}
\label{lem:ex_stopping}
Let $\tau$ be a stopping time of a sequence of random variables $(Z_t)_{t\geq 0}$.
Let $0=T_0\leq T_1\leq \cdots$ be a non-decreasing sequence and 
$\epsilon$ be a positive constant.
Suppose that $\Pr[\tau>T_\ell|\tau>T_{\ell-1}]\leq 1-\epsilon$ holds for all $\ell\geq 1$.
Then, $\E[\tau]\leq \sum_{\ell=0}^\infty(T_{\ell+1}-T_\ell)(1-\epsilon)^\ell$.
\end{lemma}
\begin{proof}
From the assumption, 
\begin{align*}
    \Pr[\tau>T_\ell]
    &= \Pr[\tau>T_\ell,\tau>T_{\ell-1}]
    =\Pr[\tau>T_\ell|\tau>T_{\ell-1}]\Pr[\tau>T_{\ell-1}]\\
    &\leq (1-\epsilon)\Pr[\tau>T_{\ell-1}]
    \leq \cdots \leq (1-\epsilon)^\ell.
\end{align*}
holds for any $\ell\geq 1$.
The first equality follows since  $\tau$ is a stopping time.
Hence, we obtain
\begin{align*}
\E[\tau]
&=\sum_{\ell=0}^\infty\sum_{t=T_\ell}^{T_{\ell+1}-1}\Pr[\tau>t]
\leq \sum_{\ell=0}^\infty\sum_{t=T_\ell}^{T_{\ell+1}-1}\Pr[\tau>T_{\ell}]
\leq  \sum_{\ell=0}^\infty(T_{\ell+1}-T_\ell)(1-\epsilon)^\ell.
\end{align*}
\end{proof}

\section{Other tools}
\begin{lemma}[Lemmas 4.24 and 4.25 in \cite{AF02}]
\label{lem:hitbounds}
Suppose that $P$ is irreducible and reversible.
Then, 
$
    \frac{1}{1-\lambda_2(P)}\leq \thit(P)\leq \frac{2}{\pi_{\min}(1-\lambda_2(P))}
$ holds.
\end{lemma}
%
%
\begin{lemma}[The Chernoff inequality (see, e.g., Theorem 1.10.21 in \cite{Doerr18})]
\label{lem:Chernoff}
Let $X_1, X_2, \ldots, X_n$ be $n$ independent random variables taking values in $[0,1]$. Let $X=\sum_{i=1}^nX_i$. 
Let $\mu^-\leq \E[X]\leq \mu^+$.
Then, 
\begin{align*}
&\Pr\left[X\geq (1+\epsilon)\mu^+\right]\leq \exp\left(-\frac{\min\{\epsilon,\epsilon^2\}\mu^+}{3}\right)
\textrm{ for any $\epsilon\geq 0$}, \\
&\Pr\left[X\leq (1-\epsilon)\mu^-\right]\leq \exp\left(-\frac{\epsilon^2\mu^-}{2}\right)
\textrm{ for any $0\leq \epsilon\leq 1$.}
\end{align*}
\end{lemma}
The following is well known as the Cheeger inequality for reversible Markov chains.
\begin{lemma}[See, e.g., Theorem 13.10 in \cite{LP17}]
\label{lem:Cheeger}
Let $P\in [0,1]^{V\times V}$ be an irreducible and reversible transition matrix.
For $S\subseteq V$, let $\pi(S)\defeq \sum_{v\in S}\pi(v)$ and $Q(S)\defeq \sum_{u\in S}\sum_{v\notin S}\pi(u)P(u,v)$.
Let $\Phi_\star\defeq \min_{S: 0<\pi(S)\leq 1/2}\frac{Q(S)}{\pi(S)}$.
Then, $\frac{\Phi_\star^2}{2}\leq 1-\lambda_\star(P)\leq 2\Phi_\star$.
\end{lemma}

\begin{lemma}[\cite{SZ19}]
\label{lem:mix_teq_SZ}
Let $x:\mathbbm{N}\to \mathbbm{R}_{>0}$ be a positive non-increasing function.
Suppose that $x(t+1)\leq x(t)\left(1-\frac{x(t)}{K}\right)$ holds for any $t\geq 0$.
Then, $x(t)\leq 1$ holds for any $t\geq \frac{\mathrm{e}^2}{\mathrm{e}-1}K+\log x(0)+1$.
\end{lemma}
\begin{proof}
First, we show that $x\left(a+\left \lceil \mathrm{e}K/x(a) \right \rceil\right)\leq x(a)/\mathrm{e}$ holds for any $a\geq 0$ by contradiction: 
Suppose that $x(a+\ell)>x(a)/\mathrm{e}$ holds for $\ell=\left \lceil \mathrm{e}K/x(a) \right \rceil$.
Since $x$ is non-increasing and $x(t+1)\leq x(t)\left(1-\frac{x(t)}{K}\right)$ holds, we have
\begin{align}
x(a+\ell)
&\leq x(a+\ell-1)\left(1-\frac{x(a+\ell-1)}{K}\right) 
\leq x(a+\ell-1)\left(1-\frac{x(a+\ell)}{K}\right) \nonumber \\
&\leq \cdots \leq x(a)\left(1-\frac{x(a+\ell)}{K}\right)^\ell
\leq x(a)\exp\left(-\frac{\ell x(a+\ell)}{K}\right)\nonumber \\
&\leq x(a)\exp\left(-\frac{\mathrm{e} x(a+\ell)}{x(a)}\right)
<x(a)/\mathrm{e}. \nonumber 
\end{align}
This contradicts the assumption and we obtain the claim: $x(a+\ell)\leq x(a)/\mathrm{e}$ holds.
Now, let $\ell(1)=\left \lceil \frac{\mathrm{e}K}{x(0)} \right \rceil$ and $\ell(i)=\left \lceil \frac{\mathrm{e}K}{x(L(i-1))} \right \rceil$ for $i\geq 1$, where $L(i)\defeq \sum_{j=1}^{i}\ell(j)$ for $i\geq 1$ and $L(0)\defeq 0$.
From the above argument, we have $x(L(i+1))\leq x(L(i))/\mathrm{e}$  for any $i\geq 0$.
Let $H$ be the first number with $x(L(H))\leq 1$.
Since
$
1<x(L(H-1))\leq x(L(H-2))/\mathrm{e}\leq \cdots \leq x(0)/\mathrm{e}^{H-1}
$ holds, 
we have $H-1<\log x(0)$. Hence, we have
\begin{align*}
L(H)&=\sum_{i=1}^H\ell(i)
=\sum_{i=0}^{H-1}\left \lceil \frac{\mathrm{e}K}{x(L(i))} \right \rceil
\leq \sum_{i=0}^{H-1}\left \lceil \frac{\mathrm{e}K}{\mathrm{e}^{H-1-i}x(L(H-1))} \right \rceil
\leq  \sum_{i=0}^{H-1}\left \lceil \frac{\mathrm{e}K}{\mathrm{e}^{i}} \right \rceil\\
&\leq H+\frac{\mathrm{e}^2}{\mathrm{e}-1}K \leq \log x(0)+1+\frac{\mathrm{e}^2}{\mathrm{e}-1}K.
\end{align*}
\end{proof}

\end{document}